\documentclass[a4paper, USenglish, autoref, cleveref, thm-restate, numberwithinsect,]{lipics-v2021}
\usepackage[T1]{fontenc}
\usepackage{algorithm, algpseudocode, braket, complexity}
\usepackage{amsmath, amssymb, amsfonts, hyperref, mathtools, tikz}
\tikzset{MyNode/.style={circle, draw, inner sep=2,outer sep=0, fill=gray}}

\newcommand{\spn}{\text{span}}
\newcommand{\rk}{\text{rk}}
\newcommand{\mydiamond}{\rotatebox[origin=c]{45}{$\vcenter{\hbox{$\Box$}}$}}
\newenvironment{subproof}[1]{\par\noindent \textit{Proof of #1.}\ }{\hfill \mydiamond \par\vspace{11pt}}

\crefname{claim}{Claim}{Claims}

\title{A combinatorial framework for clustering graph states: Algorithms and hardness for rank-integrity}
\titlerunning{A combinatorial framework for clustering graph states}

\author{Romain Bourneuf}{Univ. Bordeaux, CNRS, Bordeaux INP, LaBRI, UMR 5800, F-33400 Talence, France}{}{}{}
\author{Nathan Claudet}{University of Innsbruck, Department of Theoretical Physics, Technikerstraße  21a, A-6020 Innsbruck, Austria}{}{}{}
\author{Sang Yoon Kim}{School of Computer Science, Georgia Institute of Technology, USA}{}{}{}
\author{Rose McCarty}{School of Computer Science and School of Mathematics, Georgia Institute of Technology, USA}{}{}{} 
\author{Blair D. Sullivan}{Collegium de Lyon, ENS de Lyon, LIP, France \and Kahlert School of Computing, University of Utah, USA}{}{}{}
\author{St\'{e}phan Thomass\'{e}}{Univ. Lyon, ENS de Lyon, UCBL, CNRS, LIP, France}{}{}{}
\authorrunning{R. Bourneuf, N. Claudet, S.Y. Kim, R. McCarty, B. Sullivan, and S. Thomass\'{e}} 


\hideLIPIcs  
\nolinenumbers 
\pdfoutput=1 

\EventEditors{John Q. Open and Joan R. Access}
\EventNoEds{2}
\EventLongTitle{42nd Conference on Very Important Topics (CVIT 2016)}
\EventShortTitle{CVIT 2016}
\EventAcronym{CVIT}
\EventYear{2016}
\EventDate{December 24--27, 2016}
\EventLocation{Little Whinging, United Kingdom}
\EventLogo{}
\SeriesVolume{42}
\ArticleNo{23}

\begin{document}

\maketitle

\begin{abstract}
We introduce a new notion of distance between two graph states $\ket{G}$ and $\ket{G'}$ on the same set of qubits. This distance is the minimum number of ancilla qubits in a graph state $\ket{\widehat{G}}$ from which both $\ket{G}$ and $\ket{G'}$ can be ``easily prepared''. (When preparing graph states, we are only allowed to use one-qubit Clifford gates, one-qubit Pauli measurements, and classical communication.) We give a graphical description of this distance through the lens of vertex-minors. We then show how this distance yields quantum network analogs of many graph edit-distance problems. 

Using this framework, we develop classical algorithms for identifying the ``highly entangled clusters'' of a graph state $\ket{G}$. The \textsc{ancilla integrity} problem asks, given a graph $G$ and integer $k$, for the minimum -- over all graph states $\ket{G'}$ with distance at most $k$ from $\ket{G}$ -- of the maximum component size of $G'$. Up to a factor of $2$ in the number of ancilla qubits, this problem is equivalent to \textsc{rank integrity}, where the distance between $G$ and $G'$ is instead the minimum rank of the sum of their adjacency matrices over $\text{GF}(2)$. We prove that \textsc{rank integrity} is \XP{} parameterized by $k$. We also prove the complementary hardness result that \textsc{rank integrity} is \W[1]-hard in $k$. Finally, we give an explicit $\mathcal{O}(n^6)$-time algorithm for \textsc{ancilla integrity} when $G$ has $n$ vertices and $k=1$.
\end{abstract}

\section{Introduction}
\label{sec:intro}

There is a vast literature of \emph{graph edit distance} problems where the goal is to transform a given graph $G$ into a member of some target graph class by making as few changes as possible. The most commonly studied changes are vertex and edge deletion and insertion. This framework is very flexible and also includes other \emph{graph modification} and \emph{connectivity augmentation} problems such as those studied in~\cite{CR26, KT26}. How much more ``connectivity'' can be added to a road network by building a few more lanes? How much can be destroyed by putting a few lanes out of service? In this manner, graph edit distance problems can be used to ``cluster'' a graph into ``highly connected pieces'' and to compute its ``resiliency'', ``robustness'', ``vulnerability'', and so on.

Many different notions of edit distance, many different target graph classes, and many different measures of connectivity have been considered; see~\cite{2023editDistSurvey, 2010editDistSurvey, GroheSimilarity25} for surveys. Edit distance problems such as \textsc{feedback vertex set} -- can $k$ vertices be deleted from $G$ to obtain a forest? -- also form the foundation of modern parameterized complexity. 

In this paper we introduce a new notion of edit distance for graphs/graph states which is motivated by quantum networking. It captures whether a network provider can use a few ancilla qubits to transform one distributed quantum network into another. To explain this idea, consider two $n$-qubit graph states $\ket{G}$ and $\ket{G'}$ on the same set of qubits $V$. We consider $\ket{G}$ and $\ket{G'}$ to be ``similar'' if a network provider can prepare a graph state $\ket{\widehat{G}}$ on the disjoint union of $V$ and a small set of ancilla qubits $A$ so that the following holds. When the network provider keeps $A$ and distributes $V$ among $n$ different parties, everyone can collaboratively work together to prepare whichever of $\ket{G}$, $\ket{G'}$ is desired (without requiring any additional entanglement). Thus, by consuming a few ancilla qubits, we can choose whether to prepare $\ket{G}$ or $\ket{G'}$ on the fly, that is, \emph{after} the network has been distributed.

More formally, we define the $\emph{distance}$ $d(\ket{G}, \ket{G'})$ between two graph states on the same vertex set to be the minimum number of ancilla qubits in a graph state $\ket{\widehat{G}}$ from which each of $\ket{G}$ and $\ket{G'}$ can be prepared deterministically using only one-qubit Clifford gates, one-qubit Pauli measurements\footnote{We consider measurements to be \emph{destructive}; so after a qubit is measured it is removed from the system.}, and classical communication. We add these additional restrictions on the allowable operations in order to make the problem more manageable\footnote{We conjecture that, in a suitable multiparty model where only one-qubit Pauli measurements, local Clifford gates, and classical communication are allowed, it does not make a difference whether the ancilla qubits $A$ are owned by one party or by different parties. Here, a Clifford gate is \emph{local} if it only acts on qubits owned by a single party; so in particular, if the ancilla qubits are all owned by one party, then entangling Clifford gates may be applied on them. However, in our definition, one should consider the ancilla qubits to be owned by $|A|$ distinct parties. We refer the reader to~\cite{bravyi2024generating} for an application where the multiparty model offers an advantage for preparing particular graph states.}.
With this definition in hand, we have a new notion of ``edit distance'' for quantum graph states. 

We can obtain a completely graphical characterization of this new distance by applying some known results about graph states~\cite{transformingStates, VandenNest04, Hein04}. This graphical characterization agrees with the notion of ``perturbations'' which was recently introduced by~\cite{campbell2026erdHos} in the context of vertex-minors. The \emph{vertex-minors} of a graph $G$ are the graphs that can be obtained from $G$ by taking induced subgraphs (i.e. deleting vertices) and by performing certain operations called local complementations; \emph{locally complementing} at a vertex $v$ of $G$ replaces the induced subgraph on the neighborhood of $v$ by its complement. That is, it ``switches'' edges and non-edges within the neighborhood of $v$. Given two $n$-vertex graphs $G$ and $G'$ on the same vertex set, we define the $\emph{distance}$ $d(G, G')$ as the minimum integer $k$ so that there exists a graph on $n+k$ vertices which contains both $G$ and $G'$ as vertex-minors. Then a graph $G'$ is a \emph{$k$-perturbation} of a graph $G$ if $d(G, G')\leq k$.

It follows from~\cite{transformingStates, VandenNest04, Hein04} that these two notions of distance coincide.

\begin{restatable}{observation}{obsDistance}
\label{obs:distance}
    For any graphs $G$ and $G'$ on the same vertex set, $d(\ket{G}, \ket{G'})=d(G, G')$.
\end{restatable}

Local complementation is well-studied in structural graph theory~\cite{BouchetCircleChar, RMgrid, mccarty2021local}. It arose from attempts to generalize the graph minors theory of Robertson and Seymour~\cite{graphMinors20WQO} to the setting of symmetric binary matrices (see~\cite{graphicIsoSystems, deltaMatroidsSurvey}). Essentially, the goal is to understand the structure of graphs through their adjacency matrices instead of their incidence matrices. However, local complementation is still rather abstract and can be difficult to work with. So we now show that for edit-distance problems, usually a simpler notion called ``flipping'' suffices. In this manner, we connect the recently popular operation of ``flipping'' with quantum networking. 

To perform a \emph{flip} on a graph $G$, we select an arbitrary set of vertices $S$ of $G$, and replace the induced subgraph on $S$ by its complement. That is, we ``switch'' edges and non-edges within $S$; see~\Cref{fig:flipping}. Note that flipping on $S$ is the same thing as adding a new vertex $v$ with neighborhood $S$, locally complementing at $v$, and then deleting~$v$. This operation is also sometimes called \emph{partial complementation} or \emph{subgraph complementation}.

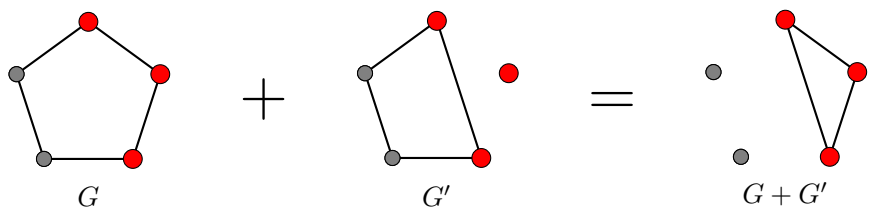
\begin{figure}
\centering
\begin{tikzpicture}[scale = .5, every node/.style={MyNode}]
    \def \r {2}
    \def \start {90}
    \foreach \i in {1, ..., 5}{
        \pgfmathparse{int(\i-1)}\edef\last{\pgfmathresult}
        \node (X\last) at (\start+\i*72-72:\r) {};
    }
    \foreach \i in {0, ..., 3}{
        \pgfmathparse{int(\i+1)}\edef\next{\pgfmathresult}
        \draw[thick] (X\i) -- (X\next);
    }
    \draw[thick] (X4) -- (X0);
    \node[draw=none, fill=none] at (0, -2.2) {};
    \node[rectangle, draw=none, fill=none] at (0, -2.6) {{\large $G$}};
    \node[fill=red, inner sep=2.5] at (X0) {};
    \node[fill=red, inner sep=2.5] at (X4) {};
    \node[fill=red, inner sep=2.5] at (X3) {};
    \node[] at (X2) {};
    \node[] at (X1) {};
\end{tikzpicture}\hskip.75cm\begin{tikzpicture}[scale=.5, every node/.style={MyNode}]
    \def \r {2}
    \node[draw=none, fill=none] at (90:\r) {};
    \node[rectangle, draw=none, fill=none] at (0, -2.05) {\textcolor{white}{\large $G*v$}};
    \node[rectangle, draw=none, fill=none] (center) at (0, .5) {\Huge$+$};
\end{tikzpicture}\hskip.75cm\begin{tikzpicture}[scale = .5, every node/.style={MyNode}]
    \def \r {2}
    \def \start {90}
    \foreach \i in {1, ..., 5}{
        \pgfmathparse{int(\i-1)}\edef\last{\pgfmathresult}
        \node (X\last) at (\start+\i*72-72:\r) {};
    }
    \foreach \i in {0, ..., 2}{
        \pgfmathparse{int(\i+1)}\edef\next{\pgfmathresult}
        \draw[thick] (X\i) -- (X\next);
    }
    \draw[thick] (X3) -- (X0);
    \node[fill=red, inner sep=2.5] at (X0) {};
    \node[fill=red, inner sep=2.5] at (X4) {};
    \node[fill=red, inner sep=2.5] at (X3) {};
    \node[] at (X2) {};
    \node[] at (X1) {};
    \node[draw=none, fill=none] at (0, -2.2) {};
    \node[rectangle, draw=none, fill=none] at (0, -2.6) {{\large $G'$}};
\end{tikzpicture}\hskip.75cm\begin{tikzpicture}[scale=.5, every node/.style={MyNode}]
    \def \r {2}
    \node[draw=none, fill=none] at (90:\r) {};
    \node[rectangle, draw=none, fill=none] at (0, -2.05) {\textcolor{white}{\large $G*v$}};
    \node[rectangle, draw=none, fill=none] (center) at (0, .4) {\Huge$=$};
\end{tikzpicture}\hskip.75cm\begin{tikzpicture}[scale = .5, every node/.style={MyNode}]
    \def \r {2}
    \def \start {90}
    \foreach \i in {1, ..., 5}{
        \pgfmathparse{int(\i-1)}\edef\last{\pgfmathresult}
        \node (X\last) at (\start+\i*72-72:\r) {};
    }
    \draw[thick] (X3) -- (X4);
    \draw[thick] (X4) -- (X0);
    \draw[thick] (X0) -- (X3);
    \node[draw=none, fill=none] at (0, -2.2) {};
    \node[rectangle, draw=none, fill=none] at (0, -2.6) {{\large $G + G'$}};
    \node[fill=red, inner sep=2.5] at (X0) {};
    \node[fill=red, inner sep=2.5] at (X4) {};
    \node[fill=red, inner sep=2.5] at (X3) {};
    \node[] at (X2) {};
    \node[] at (X1) {};
\end{tikzpicture}
\caption{A graph $G'$ obtained from $G$ by flipping, as well as the graph $G+G'$.}
\vspace{-.4cm}
\label{fig:flipping}
\end{figure}

Various ``flip-based'' graph parameters have recently been used to efficiently decide first-order properties of graphs~\cite{mergeWidth25, flipperGamesMonStable23, flipWidth23}. Indeed, this area has many deep connections to model theory which arise from the fact that performing a bounded number of flips is a quantifier-free transduction; see~\cite{transductionSurvey26, PilipczukSurvey26}. There is also a somewhat distinct line of recent research which studies the time complexity of flipping to a target graph class; see~\cite{partialComplHFree22, partialComplAPS25, partialComplHardAPS24, partialComplementFGST20}. We use \cref{obs:distance} to prove that if the target graph class is closed under local complementation, then these problems have a nice quantum interpretation. Sen, Goodenough, and Towsley~\cite{SGT23} used a similar flip-based approach to distribute a graph state across a quantum network. However, the connection to graph edit distance problems seems to be unexplored.

Let us explain this connection through the key example of \textsc{ancilla integrity}. The \emph{$k$-ancilla-integrity} of a graph state $\ket{G}$ is the minimum -- over all graph states $\ket{G'}$ with $d(\ket{G}, \ket{G'})\leq k$ -- of the maximum component size of $G'$. (The \emph{size} of a component is the number of vertices it contains.) The \textsc{ancilla integrity} problem asks, given a graph $G$ and integer $k$, to compute the $k$-ancilla-integrity of $G$. The components of $G'$ can naturally be interpreted as the ``highly entangled clusters'' of $G$; they form the optimal way to ``disentangle'' the qubits of $\ket{G}$ using only $k$ ancilla qubits. Indeed, the $k$-ancilla-integrity of $\ket{G}$ is small if and only if $G$ can ``almost'' be written as a tensor product of graph states which each only have a small number of qubits. We show in \cref{sec:quantum} how this notion relates to another previously studied parameter of graph states called Pauli persistency~\cite{BRpersistency2001,Hein04}. 

Now, consider again the example in \cref{fig:flipping} of a graph $G'$ which is obtained from $G$ by flipping. In general, given graphs $G$ and $G'$ on the same vertex set $V$, we write $G+G'$ for the graph on $V$ whose edge set is the symmetric difference of the edge sets of $G$ and $G'$. So in \cref{fig:flipping}, this graph $G+G'$ is a clique plus some isolated vertices. The adjacency matrix of a clique has minimum rank $1$; the \emph{minimum rank} of a graph is the minimum -- over all ways of filling in the diagonal of the adjacency matrix with $0$s and $1$s -- of its rank over $\text{GF}(2)$. It is almost true that for any graphs $G$ and $G'$ on the same vertex set, $G'$ can be obtained from $G$ by flipping on at most $k$ sets if and only if $G+G'$ has minimum rank at most $k$. The only difference is that sometimes one additional flip is required~\cite{Aminrank75}; see also~\cite{BPRminrank22} and~\cite[Lemma~4.1]{davies2025preparing}. Since there are other nice connections between rank in the adjacency matrix and entanglement in the graph state (see~\cite{Hein04, vdNDVB07}), we focus on the rank-based definition\footnote{We note that it is also natural to consider adding a low-rank symmetric matrix as a good notion of edit distance from other perspectives; see for instance~\cite{FGP2020,MMS16} for more about the matrix-based viewpoint.}.

So, we define the \emph{$k$-rank-integrity} of a graph $G$ to be the minimum -- over all graphs $G'$ so that $G+G'$ has minimum rank at most $k$ -- of the maximum component size of $G'$. We call any such graph $G'$ a \emph{rank-$k$ perturbation} of $G$. Thus a rank-$k$ perturbation of $G$ is obtained from $G$ by adding a symmetric binary matrix of rank at most $k$ to its adjacency matrix, and then setting the diagonal to be $0$. The \textsc{rank integrity} problem asks, given a graph $G$ and integer $k$, to compute the $k$-rank-integrity of $G$. We prove the following lemma in \cref{sec:quantum}. 

\begin{restatable}{lemma}{flipAncillaComparison}
\label{lem:flipAncillaComparison}
For any graph $G$ and non-negative integer $k$, \begin{align*}
(2k)\textit{-rank-integrity}(G)\leq k\textit{-ancilla-integrity}(\ket{G})\leq k\textit{-rank-integrity}(G).
\end{align*}
\end{restatable}

Thus, rank integrity and ancilla integrity capture the same underlying notion, up to a change in the parameter. Accordingly, translating between the two notions incurs at most a factor-of-two change in the parameter $k$. Importantly, the factor of two appears only in the parameter, not in a direct comparison of the resulting integrity values. 

We note that the same relationship holds for any target graph class or target graph parameter which is invariant under local complementation. (Note that the maximum component size of a graph is invariant under local complementation.) There are many such invariants since cut-rank (or equivalently, Schmidt-rank of the corresponding graph state~\cite[Proposition~10]{Hein06}) is invariant under local complementation; see~\cite{connectivityIsotropic,Hein04,RWAndVM}. Many such graph parameters also have nice quantum interpretations, including twin classes~\cite{foliagePartition}, rank-width~\cite{vdn2006universal,vdNDVB07}, linear rank-width~\cite{photonicLinearRW}, and minimum degree up to local complementation~\cite{CattaneoP15, Javelle12}. 

Our first main result (and the heart of this paper) is the following theorem.

\begin{restatable}{theorem}{thmAlgorithm}
\label{thm:algo-XP}
    \textsc{rank integrity} is \XP{} parameterized by~$k$.
\end{restatable}

That is, there is an algorithm which takes as input an $n$-vertex graph $G$ and integer $k$ and computes the $k$-rank-integrity of $G$ in time $n^{f(k)}$ for some computable function $f$. 

Even when $k=1$ and we are looking for a flip $G'$ of $G$, the naive approach is to iterate over all $2^n$ subsets of the vertex set; obtaining a polynomial bound is surprisingly nontrivial. However, note that for graphs $G$ where the optimal flip $G'$ of $G$ has exactly two components, we are equivalently looking for a ``most balanced split'' of $G$. A \emph{split} is a partition of the vertex set of a graph into two parts so that the edges between them form the edge-set of a complete bipartite graph. It is \emph{most balanced} if the two parts are as equal in size as possible. Efficient algorithms to find the ``split decomposition'' of a graph are well-known~\cite{linearSplitRevisited12, cunningham82, linearSplitDahlhaus, spinrad89}. The split decomposition is a tree that ``displays'' all of the splits of a graph simultaneously. So, at an intuitive level, the main idea of our algorithm is to first reduce to the two component case, and then to do some dynamic programming using the split decomposition.

The approach for general $k$ is somewhat similar, however we do not have access to something like the split decomposition which displays all of the relevant information at once. To explain our approach, let us say that a \emph{cut} of a graph $G$ is a partition $(A,B)$ of its vertex set into two parts. The \emph{rank} (or \emph{cut-rank}) of $(A,B)$ is the rank over $\text{GF}(2)$ of the submatrix of the adjacency matrix with rows $A$ and columns $B$. The \textsc{most balanced cut-rank} problem asks, given a graph $G$ and integer $k$, for the most balanced cut of rank at most $k$. Very recently, Boja\'{n}czyk, Mi. Pilipczuk, Przybyszewski, Soko{\l}owski, and Stamoulis~\cite[Theorem~6.1]{BPPSS25} proved a theorem which implies that \textsc{most balanced cut-rank} is $\XP$ in $k$. (This theorem has also been generalized to arbitrary connectivity functions~\cite{OS26}.) Essentially, this theorem shows that every cut of rank at most $k$ is displayed in one of polynomially-many (in the number of vertices $n$) partitions of the vertex set. 

Thus our main contribution is handling the case where the optimal rank-$k$ perturbation $G'$ of $G$ has many components. We note that even when $G'$ has exactly three components, it is not clear how to use the analysis in~\cite{BPPSS25}. As such, our proof of \Cref{thm:algo-XP} is self-contained, although our methods are inspired by~\cite{BPPSS25}. We note that quantum algorithms for \textsc{most balanced cut-rank} have been considered in~\cite{BGW25}, as have other heuristic approaches~\cite{PHSH26}. 

Complementing \Cref{thm:algo-XP}, we prove the following hardness result.

\begin{restatable}{theorem}{thmHardness}\label{thm:W1-hard}
    \textsc{rank integrity} is \W$[1]$-hard parameterized by~$k$.
\end{restatable}

That is, we cannot expect to find an algorithm which takes in an $n$-vertex graph $G$ and an integer $k$ and computes the $k$-rank-integrity of $G$ in time $f(k)n^{c}$ for some fixed constant $c$ and computable function $f$. (Doing so is at least as hard as determining if an $n$-vertex graph has a clique of size $k$.) We reduce directly from \textsc{clique}. Our reduction is inspired by the proof in~\cite{DDV16} that \textsc{component order connectivity} is \W$[1]$-hard parameterized by~$k$. In this problem, given a graph $G$ and integer $k$, we wish to compute the $k$-order-integrity of $G$. The \emph{$k$-order-integrity} of $G$ is the minimum -- over all sets $X \subseteq V(G)$ of size at most $k$ -- of the maximum component size of $G-X$.

We note that rank-integrity is really a ``dense analog'' of order-integrity in that for graphs without a $K_{t,t}$-subgraph, they only differ in a change of the parameter (see \Cref{lem:denseAnalog}). This is an instance of a more general trend in the area of ``structural sparsity''; we refer the reader to the recent survey of Mi. Pilipczuk~\cite{PilipczukSurvey26} for more examples. We note that \Cref{lem:denseAnalog} does not seem strong enough to obtain \Cref{thm:W1-hard} directly; we will have to work harder.

Finally, we give some evidence that computing \textsc{rank integrity} could be a useful intermediate step towards obtaining an exact \XP{} algorithm for \textsc{ancilla integrity}. We take this approach to obtain an explicit algorithm for computing $1$-ancilla integrity.

\begin{restatable}{theorem}{thmAncillaVuln}
\label{thm:ancillaVuln}
    The $1$-ancilla-integrity of an $n$-vertex graph $G$ can be computed in time~$\mathcal{O}(n^6)$.
\end{restatable}

We conjecture that, like \textsc{rank integrity}, the \textsc{ancilla integrity} problem is $\XP$ and $\W[1]$-hard in~$k$. We note that this combinatorial framework for edit distance problems on quantum networks leads to a whole host of interesting questions. For instance, one can ask the following connectivity augmentation problem; what is the minimum integer $k$ so that $G$ has a rank-$k$ perturbation which is prime? (A graph is \emph{prime} if it has no non-trivial splits.) From the matrix viewpoint it is also natural to ask whether the following problem is $\XP$ in $k$; given an $n \times m$ binary matrix $M$, what is the minimum -- over all $n \times m$ binary matrices $N$ of rank at most $k$ -- of the maximum block size of $M+N$? (The \emph{size} of a block is its number of rows plus columns.) It also seems unlikely that replacing ``maximum component size'' with ``average component size'' could change the complexity of the \textsc{rank integrity} problem much. Can we capture all of these problems in one language and prove that every problem which can be stated in that language is polynomial-time solvable?

\section{Proof overviews}

In this section we overview our approaches to proving Theorems~\ref{thm:algo-XP}, \ref{thm:W1-hard}, and~\ref{thm:ancillaVuln}. Their full proofs appear in Sections~\ref{sec:XPalgorithm}, \ref{sec:hardness}, and~\ref{sec:ancillaAlgorithm}, respectively. Here, however, we overview the proofs in the reverse order; so we begin with \Cref{thm:ancillaVuln} about 1-ancilla-integrity, then \Cref{thm:W1-hard} about $\W[1]$-hardness, and finally \Cref{thm:algo-XP} about the $\XP$ algorithm.

\subsection{Overview of the algorithm for 1-ancilla-integrity}

We begin by discussing how to compute the $1$-ancilla-integrity of an $n$-vertex graph $G$. 

First of all, let $G^*$ be a $1$-perturbation of $G$ which minimizes the maximum component size. Recall from the definition of a $1$-perturbation that $G$ and $G^*$ have the same vertex set, and both are vertex-minors of a graph $\widehat{G}$ on $n+1$ vertices. Let $a$ be the extra vertex, that is, the unique vertex in $V(\widehat{G})\setminus V(G)$. By locally complementing at vertices of $\widehat{G}$, we may assume that $G$ is an induced subgraph of $\widehat{G}$. We now use the following lemma.

\begin{lemma}[{Bouchet~\cite[9.2]{graphicIsoSystems}} and {Fon-Der-Flaass~{\cite[Corollary~4.3]{FonDerFlaass1988}}}]
\label{lem:threeWaysToRemove}
    For any graph $G^*$ which is a vertex-minor of a graph $\widehat{G}$ and any vertex $a \in V(\widehat{G})\setminus V(G^*)$, the graph $G^*$ is also a vertex-minor of at least one of the following three graphs:\begin{enumerate}
\item the graph obtained from $\widehat{G}$ by deleting $a$,
\item the graph obtained from $\widehat{G}$ by locally complementing at $a$ and then deleting $a$, or
\item the graph obtained from $\widehat{G}$ by selecting an arbitrary neighbor $u$ of $a$, locally complementing on $u$ then $a$ then $u$ again, and finally deleting $a$.
\end{enumerate}
\end{lemma}

Since locally complementing does not change the components of a graph, we may assume that $G^*$ actually equals one of the three graphs above. In this manner, we prove in \Cref{lem:reductionAncilla} that the $1$-ancilla-integrity of $G$ equals the minimum rank-$1$ integrity of any graph which is obtained from $G$ by performing at most one local complementation. Thus we just need to show how to compute the rank-$1$ integrity. Since a graph $G'$ is a rank-$1$ perturbation of $G$ if and only if it is a flip of $G$, we call this the \emph{flip-integrity} of $G$. 

Given a set $S$ of vertices of $G$, we write $G\circ S$ for the graph obtained from $G$ by flipping on $S$. Let $S^*$ be a set of vertices so that $G \circ S^*$ minimizes the maximum component size. In \Cref{lem:makeGConn}, we show how to reduce to the case that $G$ is connected. Now, consider the components $C_1, C_2, \ldots, C_r$ of the graph $G \circ S^*$. The high-level approach is to solve the problem in two stages; first we solve the problem in the case that $r\geq 3$, then we solve the problem for the case that $r=2$, and finally we take the best of the two solutions. (We note that $r\neq 1$ since there is always a flip that makes one vertex isolated; see \Cref{lem:lessThanN}.)

So, first suppose that $G \circ S^*$ has at least three components. Then, since $G$ is connected, $S^*$ contains at least one vertex from each of these components. So $G$ contains a triangle whose vertices $x$, $y$, and $z$ are in three different components of $G \circ S^*$. Notice that every vertex in $S^*$ forms a triangle (in $G$) with at least two of the vertices from the triangle $x,y,z$. In fact this characterizes $S^*$; no vertex outside of $S^*$ forms such a triangle in $G$. So, by iterating over all triangles $x,y,z$ in $G$ and deriving the set \begin{align*}
    S_{x,y,z}=\{u \in V(G) : u \text{ is in a triangle in $G$ with at least two of the vertices }x,y,z\},
    \end{align*} we will at some point consider $S^*$. So the first stage of our algorithm is to compute the minimum -- over all such sets $S_{x,y,z}$ -- of the maximum component size of $G \circ S_{x,y,z}$. This stage of the algorithm is formalized in \Cref{lem:constrainedFlipVul}.

In the second stage, which is formalized in \Cref{thm:compute-most-balanced-split}, our algorithm finds the most balanced split of an $n$-vertex graph $G$ in time $\mathcal{O}(n^2)$. Note that if $G \circ S^*$ has exactly two components, then those two components are the sides of a split. Conversely, if $A$ and $B$ are the sides of a split, then there is a flip so that every component has size at most the maximum of $|A|$ and $|B|$. This stage uses known linear-time algorithms for computing the ``split decomposition'' of a graph~\cite{linearSplitRevisited12, linearSplitDahlhaus}, which ``display'' all of its splits.

Finally, we just return the better of the two values from the two stages of the algorithm. The details and running-time are given in \Cref{sec:ancillaAlgorithm}.

\subsection{Overview of the proof of \texorpdfstring{\W$[1]$}{W[1]}-hardness of \textsc{rank integrity}}

As we already discussed, the problem \textsc{rank integrity} is a ``dense analog'' of \textsc{component order connectivity}. Recall that in the latter problem, given a graph $G$ and an integer $k$, the goal is to find a set $X$ of at most $k$ vertices minimizing the maximum size of a component of $G-X$.
Drange, Dregi and van ’t Hof~\cite{DDV16} proved that this problem is \W$[1]$-hard parameterized by $k$.
By adapting their reduction to the dense setting, we show that \textsc{rank integrity} is \W$[1]$-hard parameterized by $k$ as well.

Their hardness proof proceeds by a reduction from the problem \textsc{clique}, parameterized by the size of the desired clique.
Given an instance $(G, k)$ of \textsc{clique}, their reduction uses a variant of the vertex-edge incidence graph of $G$.
More precisely, they consider the graph $I(G)$ with vertex set $V(G) \cup E(G)$, where $V(G)$ induces a clique, $E(G)$ is an independent set, and a vertex $v \in V(G)$ is adjacent to a vertex $e \in E(G)$ if and only if $v$ and $e$ are incident in $G$.

The key observation is that there is always an optimal solution $X$ for $I(G)$ consisting only of vertices from $V(G)$. Moreover, for such a set $X$, the largest component of $I(G) - X$ has size $|I(G)| - k - e(X)$, where $e(X)$ denotes the number of edges in the subgraph $G[X]$.
It follows that $G$ contains a clique of size $k$ if and only if $I(G)$ contains a set of $k$ vertices whose removal leaves all components of size at most $|I(G)| - k - \binom{k}{2}$.

This reduction relies crucially on the fact that the clique on $V(G)$ is robustly connected in the sparse setting: after deleting $k$ vertices, it still contains a component of size $n-k$. 
This property fails in the dense setting, since a complete graph can be transformed into an independent set by a rank-$1$ perturbation.
Our strategy is therefore to replace the clique in the reduction by a graph $H$ that is robustly connected in the dense setting. Ideally, $H$ would have the property that every rank-$k$ perturbation of $H$ has a component of size at least $n-k$.

We are unable to construct such a graph directly.
Instead, we use the following weaker substitute, based on Sylvester-Hadamard matrices.
We consider the graph $H$ on vertex set $V(G) \times [N]$, where $N$ is a sufficiently large integer depending only on $k$, whose adjacency matrix is obtained from an appropriate Sylvester-Hadamard matrix. 
For $v \in V(G)$, we say that a vertex of $H$ of the form $(v, \cdot)$ is a \emph{copy} of $v$.
The crucial property of this construction, which we prove in \cref{lem:exists-giant-comp}, is that every rank-$k$ perturbation of $H$ has a component of size at least $|H| - N/2$. 
In particular, such a component contains at least one copy of \emph{every} $v \in V(G)$.
This graph $H$ will play the same role as the clique in the sparse reduction: it provides a large robust component that survives all low-rank perturbations.

We then define a graph $\Gamma_k(G)$ with vertex set $(V(G) \cup E(G)) \times [N]$. The set $V(G) \times [N]$ induces the graph $H$, the set $E(G) \times [N]$ is an independent set, and vertices $(v, \cdot) \in V(G) \times [N]$ and $(e, \cdot) \in E(G) \times [N]$ are adjacent if and only if $v$ and $e$ are incident in $G$.
Thus, $\Gamma_k(G)$ can be obtained starting from the vertex-edge incidence graph of $G$, by replacing each vertex and each edge of $G$ by $N$ copies, so that the copies of the vertices in $V(G)$ induce the robust graph $H$.

The key claim, which we prove in \cref{lem:pf-reduction}, is that there is a rank-$k$ perturbation of $\Gamma_k(G)$ whose largest component has size at most $|\Gamma_k(G)| - \binom{k+1}{2} \cdot N$ if and only if $G$ contains a clique of size $k+1$. This immediately implies the \W$[1]$-hardness of \textsc{rank integrity}.
The shift from $k$ to $k+1$ reflects the fact that the vertex-edge incidence matrix of a connected graph on $k+1$ vertices has rank $k$.

First, suppose that $G$ contains a clique $C$ of size $k+1$. Then, there is a rank-$k$ perturbation of $\Gamma_k(G)$ that separates all copies of $e$ for $e \in E(C)$ from the rest of $\Gamma_k(G)$.
This is done in \cref{lem:rk-Kk1}.
The largest component then has size at most $|\Gamma_k(G)| - \binom{k+1}{2} \cdot N$.
Conversely, suppose that there is a rank-$k$ perturbation $\Gamma_k(G) + \widehat{\Gamma}$ of $\Gamma_k(G)$ whose largest component has size at most $|\Gamma_k(G)| - \binom{k+1}{2} \cdot N$.
By the defining property of $H$, there is a component of $\Gamma_k(G) + \widehat{\Gamma}$ which contains all but at most $N/2$ vertices from $V(G) \times [N]$.
Consequently, this component must miss more than $\left(\binom{k+1}{2} - 1\right) \cdot N$ vertices from $E(G) \times [N]$.
Let $F \subseteq E(G)$ be the set of edges $e \in E(G)$ such that this component misses at least one copy of $e$. The previous bound implies that $|F| \geq \binom{k+1}{2}$.
It remains to extract a clique from $G[F]$.
Observe that, for each $e \in F$, some copy of $e$ is separated from the component which contains a copy of each vertex of $G$, hence of the endpoints of $e$. Thus, the perturbation $\widehat{\Gamma}$ must cancel the corresponding incidence adjacencies, so the vertex-edge incidence matrix of $G[F]$ appears as a submatrix of the adjacency matrix of $\widehat{\Gamma}$.
In particular, this incidence matrix has rank at most $k$.
However, the vertex-edge incidence matrix of any graph with at least $\binom{k+1}{2}$ edges and no clique of size $k+1$ has rank at least $k+1$, see \cref{prop:find-Kk}.
Therefore $G[F]$ must contain a clique of size $k+1$.
Hence $G$ contains a clique of size $k+1$, as required.

\subsection{Overview of the \XP{} algorithm for \textsc{rank integrity}}
\label{subsec:XPalgorithmOverview}

As is standard for \XP~algorithms, we describe the algorithm in terms of a bounded number of ``guesses'', thus yielding a polynomial number of ``guessed objects''. Formally, each guess corresponds to a brute-force enumeration over all possibilities, and the algorithm returns the best solution found over all non-rejected branches. A roadmap of our algorithm can be found in \cref{alg:pseudocode}.

Recall that our goal is, given a graph $G$ and integer $k$, to find a rank-$k$ perturbation $G'$ of $G$ that minimizes the maximum component size of $G'$. It is convenient to introduce looped graphs to store this perturbation more succinctly; a \emph{looped graph} is a graph where every vertex is allowed to have at most one loop. Thus we can equivalently describe a rank-$k$ perturbation $G'$ of $G$ in the form $G'=G+\widehat{G}$ where $\widehat{G}$ is a looped graph of rank at most $k$; the \emph{rank} of a looped graph is the rank of its adjacency matrix over $\text{GF}(2)$. We now consider the resulting graph $G'$ to also be looped; this does not change its maximum component size.

Fix a rank-$k$ perturbation $G+G^*$ of $G$ that minimizes the maximum component size.
To design an \XP~algorithm, we would like to be able to describe $G^*$ using only a bounded number of vertices of $G$ and bounded additional information. For this, we need some definitions. 

In a looped graph $H$, the \emph{neighborhood} of a vertex $v$ is the set of vertices that are adjacent to $v$ in $H$. 
So we consider $v$ to be in its own neighborhood if and only if there is a loop at $v$.
The relation $\sim$ on $V(H)$ defined by $u \sim v$ if and only if $u$ and $v$ have the same  neighborhood in $H$ is an equivalence relation.
The partition into equivalence classes for $\sim$ is the \emph{type partition} of $H$.
The \emph{quotient graph} $H/\mathcal{P}$ is the looped graph on vertex set $\mathcal{P}$ where $P_1, P_2 \in \mathcal{P}$ are adjacent if and only if there is an edge between them in $H$ (in which case they are complete in $H$; two sets of vertices are \emph{complete} if $H$ contains all possible edges between them and \emph{anticomplete} if $H$ contains no edges between them).

Since $G^*$ has rank at most $k$, its type partition $\mathcal{P}^*$ has at most $2^k$ parts (because a binary rank-$k$ matrix has at most $2^k$ distinct columns).
Observe that $G^*$ is entirely determined by the looped graph $G^*/\mathcal{P}^*$ (which has at most $2^k$ vertices) and by the partition $\mathcal{P}^*$.
Our algorithm first guesses the looped graph $G^*/\mathcal{P}^*$, and then tries to recover enough information about $\mathcal{P}^*$ to reconstruct $G^*$.
We first explain the algorithm under the simplifying assumption that $G+G^*$ has a bounded number of components. 
This already contains the main ideas: guessing representatives, defining possible assignments, and resolving conflicts. 
We then explain how to remove this assumption.

Let $\mathcal{Q}^*$ be the common refinement of $\mathcal{P}^*$ and of the partition into components of $G+G^*$; note that $\mathcal{Q}^*$ consists of a bounded number of parts.
To describe the partition $\mathcal{P}^*$, we guess a vertex $v_Q$ from each part $Q \in \mathcal{Q}^*$, and denote the corresponding set of guessed representative vertices by $R^* \coloneqq (v_Q)_{Q \in \mathcal{Q}^*}$.
We store which vertices in $R^*$ should belong to the same component of $G+G^*$ and which vertices in $R^*$ should belong to the same part of $\mathcal{P}^*$.
Then, for every vertex $v \in V(G)$, we start trying to figure out which part $Q \in \mathcal{Q}^*$ contains $v$.
For this, let us start with an easy observation.
Let $v \in V(G)$ and let $Q \in \mathcal{Q}^*$ be the part that contains $v$.
By definition, $v$ and $v_{Q}$ belong to the same part $P \in \mathcal{P}^*$ and to the same component $K$ of $G+G^*$.
Thus, $v$ and $v_Q$ have the same neighborhood in $G^*$, so for every vertex $v_{Q'} \in R^*$ that does not belong to $K$, $v$ and $v_{Q}$ are either both adjacent or both non-adjacent to $v_{Q'}$ in $G$.
For every vertex $v_Q \in R^*$ that satisfies this, we say that $v$ and $v_{Q}$ are \emph{compatible}. 
Let $f^* : V(G) \to 2^{R^*}$ be the function which maps every vertex $v \in V(G) \setminus R^*$ to its compatible representatives, and $v_Q$ to $\{v_Q\}$ for every $v_Q \in R^*$.

Now, we know that every vertex $v \in V(G)$ must be assigned to some part $Q \in \mathcal{Q}^*$ such that $v_Q \in f^*(v)$.
Thus, reconstructing $\mathcal{Q}^*$ amounts to choosing, for every vertex $v \in V(G)$, a representative $v_Q \in f^*(v)$, subject to consistency constraints.
Indeed, it would be nice if we could make all these choices independently to obtain a partition $\mathcal{Q}'$ corresponding to some rank-$k$ perturbation of $G$.
Unfortunately, this might not be the case: there might be vertices $v, v' \in V(G)$ and vertices $v_Q \in f^*(v)$ and $v_{Q'} \in f^*(v')$ such that $v_Q$ and $v_{Q'}$ do not belong to the same component of $G+G^*$ and exactly one of $vv'$ and $v_Qv_{Q'}$ is an edge of $G$.
However, if $Q, Q' \in \mathcal{Q}^*$ are not contained in the same component in $G+G^*$ then $Q$ and $Q'$ are either complete or anticomplete in $G$.
This means that it is not possible to assign simultaneously $v$ to $Q$ and $v'$ to $Q'$.
If this happens for every $v_Q \in f^*(v)$ and $v_{Q'} \in f^*(v')$ that do not belong to the same component of $G+G^*$, we say that $v$ and $v'$ are \emph{strongly conflicting}.
If this happens for some $v_Q \in f^*(v)$ and $v_{Q'} \in f^*(v')$ that do not belong to the same component of $G+G^*$, we say that $v$ and $v'$ are \emph{conflicting}.
If they are conflicting but not strongly conflicting, we say that $v$ and $v'$ are \emph{weakly conflicting}.

It follows immediately from this definition that any two strongly conflicting vertices must belong to the same component of $G+G^*$. This is formalized in \cref{cl:conflicting-same-cc}.
The ``strongly conflicting'' relation is symmetric, so its reflexive and transitive closure forms an equivalence relation. 
Let $\mathcal{U}^*$ be the partition of $V(G)$ into equivalence classes for this relation. 
By the above observation, every $U \in \mathcal{U}^*$ is contained in a single component of $G+G^*$.

The situation is not as simple for weakly conflicting vertices, but we show that such vertices must still satisfy very restrictive conditions.
More precisely, if $v$ and $v'$ are weakly conflicting, we establish in \cref{cl:setup-weakly-conflicting} that there exist vertices $v_{Q_1}, v_{Q_2}, v_{Q'_1}, v_{Q'_2} \in R^*$ such that $f^*(v) = \{v_{Q_1}, v_{Q_2}\}$, $f^*(v') = \{v_{Q'_1}, v_{Q'_2}\}$; $v_{Q_1}$ and $v_{Q'_1}$ belong to the same component of $G+G^*$, as do $v_{Q_2}$ and $v_{Q'_2}$; and exactly one of $v_{Q_1}v_{Q'_2}$ and $v_{Q_2}v_{Q'_1}$ is an edge in $G$.
In this case, the conflict between $v$ and $v'$ is asymmetric: if $v$ is assigned to $Q_1$ then $v'$ must also be assigned to $Q'_1$ (or vice versa, depending on whether or not $vv' \in E(G)$) but the converse implication does not hold. 

Then, let $W = W(Q_1, Q_2)$ be the set of vertices $v \in V(G)$ such that $f^*(v) = \{v_{Q_1}, v_{Q_2}\}$ and $W' = W'(Q'_1, Q'_2)$ be the set of vertices $v' \in V(G)$ such that $f^*(v') = \{v_{Q'_1}, v_{Q'_2}\}$.
Consider the relation $\prec$ on $W \cup W'$ defined by $u \prec v$ if and only if $u \in Q_1 \cup Q'_1$ implies $v \in Q_1 \cup Q'_1$.
It then follows from the above discussion that $\prec$ forms a preorder on $W \cup W'$, where any element of $W$ is comparable with any element of $W'$.
If we are given one $\prec$-minimum element of $W\cup W'$ that belongs to $Q_1\cup Q'_1$, say an element of $W$, then the preorder determines, for every $v\in W'$, whether $v$ belongs to $Q'_1$ or to $Q'_2$.
This step is formalized in \cref{cl:exists-crossing-solver}.
Therefore, by simply guessing such a \emph{witness} vertex, we are able to resolve all conflicts between $W$ and $W'$.
Since there are only a bounded number of tuples $(Q_1, Q_2, Q'_1, Q'_2)$, we are able to resolve all weak conflicts by guessing only a bounded number of witness vertices.

We incorporate the information given by these guessed witnesses by deleting from $f^*(v)$ every representative that is inconsistent with the implications forced by the witnesses. 
Let $\widehat{f^*}$ be the resulting compatibility function.
The important property of $\widehat{f^*}$ is that every conflict with respect to $\widehat {f^*}$ is a strong conflict.

Say that a partition $\mathcal{C}$ of $V(G)$ is \emph{suitable} if it satisfies the following conditions. First, for every vertex $v \in V(G)$, there exists a representative $v_Q \in \widehat{f^*}(v)$ such that $v$ and $v_Q$ belong to the same part of $\mathcal{C}$. Second, every part $U \in \mathcal{U}^*$ is contained in a single part of $\mathcal{C}$. Third, whenever two representatives $v_Q, v_{Q'} \in R^*$ are required by the stored data to lie in the same component of $G+G^*$, they belong to the same part of $\mathcal{C}$.
The crucial lemma is that every suitable partition $\mathcal{C}$ can be realized as the partition into components of some rank-$k$ perturbation of $G$.
This is proved in \cref{cl:C*-cc-G+G'}.
Among all suitable partitions, a simple dynamic programming algorithm can compute the one minimizing the maximum size of a part.
Thus, the algorithm computes such a partition and returns the maximum size of a part.

Finally, we explain how to remove the assumption that $G+G^*$ has a bounded number of components.  
Define the \emph{signature} of a component to be the set of parts of $\mathcal{P}^*$ that it intersects.
Since $\mathcal{P}^*$ has a bounded number of parts, there are only a bounded number of possible signatures.
Then, the algorithm guesses a bounded number of components of $G+G^*$ that are ``characteristic'' of all signatures.
The key observation is that the components that were not guessed can be identified solely from $\mathcal{U}^*$ and the guessed data. 
The precise condition appears in \cref{cl:updated-compatibility-info}.
These classes can simply be put aside before computing $\mathcal{C}$. 
The rest of the algorithm is unchanged.

\section{The graph state picture}
\label{sec:quantum}

In quantum information, qubits are the quantum equivalent of classical bits, and an $n$-qubit quantum state is identified with a complex vector with $2^n$ coordinates. Graph states are a family of quantum states that are in one-to-one correspondence with simple undirected graphs. Given a graph $G$ defined on the vertex set $V$, the corresponding graph state is defined by $ \ket{G} = \sum_{X \subseteq V} (-1)^{e_X} \ket{1}_{X} \ket{0}_{V \setminus X}$, where $e_X$ denotes the number of edges in the subgraph induced by $X$. Graph states were originally introduced as a universal resource for measurement-based quantum computing~\cite{briegel2009measurement, raussendorf2001one, raussendorf2003measurement}. However, they have seen recent interest as a simple but useful tool in quantum networking~\cite{hahn2022phdthesis}, as they are useful for cryptographic protocols, a prime example being quantum secret sharing~\cite{Bell2014secret, gravier2013quantum, Javelle2013, Keet2010, markham2008graph}. 

Typically, a graph state is prepared by a network provider, which sends a qubit to each each party in the network. Then, only local quantum operations are allowed, for example one-qubit gates or one-qubit measurements. It is standard to consider only the so-called \emph{Clifford operations}: one-qubit Clifford gates and one-qubit Pauli measurements. These are in general simpler to implement than arbitrary local operations. Classical communication between distant parties is also allowed; it is in general needed to correct according to the (probabilistic) measurement outcomes. The graphical interpretation of these ``simple'' local operations, given by the vertex-minor formalism, gives tools to easily manipulate entanglement across the network \cite{bravyi2024generating, Cautres2024, fischer2021distributing, freund2025graph, hahn2019quantum, Mannalath2023, meignant2019distributing}, or to identify bottlenecks~\cite{Hahn2022limitations}.

\begin{proposition}[\cite{transformingStates, VandenNest04, Hein04}]
\label{prop:vertexMinorInterp}
    A graph $H$ is a vertex-minor of a graph $G$ if and only if $\ket{G}$ transforms deterministically into $\ket{H}$ via one-qubit Pauli measurements, one-qubit Clifford gates, and classical communication.
\end{proposition}

This explains why the notions of distance agree for graphs and graph states. That is, \Cref{prop:vertexMinorInterp} implies \Cref{obs:distance}, which we restate below for convenience.

\obsDistance*

We note that $d$ is a valid notion of distance. Formally, it is a pseudometric on the class of all graph states. It follows from the definition that $d(\ket{G}, \ket{G})=0$ and that $d$ is {symmetric}, that is, $d(\ket{G}, \ket{G'}) = d(\ket{G'}, \ket{G})$. Moreover, $d$ satisfies the triangular inequality.

\begin{proposition} For any three graph states $\ket{G_1}$, $\ket{G_2}$ and $\ket{G_3}$ defined on the same qubit set, $ d(\ket{G_1},\ket{G_3}) \leqslant d(\ket{G_1},\ket{G_2}) + d(\ket{G_2},\ket{G_3})$.
\end{proposition}

\begin{proof}
    Let $V$ be the vertex set of $G_1$, $G_2$ and $G_3$. By definition, both $G_1$ and $G_2$ are vertex-minors of some graph $\widehat{G}_{1,2}$ whose vertex set consists of $V$ and $d(G_1,G_2)$ additional vertices. By locally complementing, we may assume that $G_2$ is an induced subgraph of $\widehat{G}_{1,2}$. Similarly, both $G_2$ and $G_3$ are vertex-minors of some graph $\widehat{G}_{2,3}$ whose vertex set consists of $V$ and $d(G_2,G_3)$ additional vertices. Again we may assume that $G_2$ is an induced subgraph of $\widehat{G}_{2,3}$. Now let $\widehat{G}_{1,3}$ be the union of $\widehat{G}_{1,2}$ and $\widehat{G}_{2,3}$. That is, on $V$, the graph $\widehat{G}_{1,2}$ equals $G_2$, plus it has $d(G_1,G_2)+d(G_2,G_3)$ additional vertices whose adjacencies are given by $\widehat{G}_{1,2}$ and $\widehat{G}_{2,3}$, respectively. Then $\widehat{G}_{1,3}$ contains both $G_1$ and $G_3$ as vertex-minors.
\end{proof}

\subsection{Discussion on ancilla-integrity}

The \emph{$k$-ancilla-integrity} of a graph state $\ket{G}$ defined on qubit set $V$ is the minimum integer $t$ such that there exists a graph state $\ket{\widehat{G}}$ defined on qubit set $V \cup A$ such that $\widehat{G}-A=G$, and such that there exist one-qubit Pauli measurements on $A$ that deterministically map $\ket{\widehat{G}}$ to a tensor product of quantum states with at most $t$ qubits. 

This definition is equivalent to the definition given in the introduction. Indeed, the one-qubit Pauli measurements map $\ket{\widehat{G}}$ to a quantum state that is a graph state up to one-qubit Clifford gates, which do not modify entanglement (see for example \cite{Hein04}). Conversely, suppose that $\ket{G}$ and $\ket{G'}$ can both be prepared from a graph state $\ket{\widehat{G}}$ on qubit set $V \cup A$ using one-qubit Pauli measurements, one-qubit Clifford gates, and local communication. Then we may switch the order of these operations so that $\widehat{G}-A=G$ and $\ket{G'}$ can be obtained from $\widehat{G}$ by first performing one-qubit Pauli measurements and then one-qubit Clifford gates. 

We now show that the notions of ancilla integrity and rank integrity agree up to a factor of 2 in the number of ancilla qubits. We say that two graphs $G$ and $G'$ on the same vertex set are \emph{locally equivalent} if one can be obtained from the other by performing a series of local complementations.

\flipAncillaComparison*
\begin{proof}
    First let $G'$ be a rank-$k$ perturbation of $G$. Then by~\cite[Lemma 7.1]{campbell2026erdHos}, $G'$ is also a $k$-perturbation of $G$. So the $k$-ancilla-integrity of $G$ (or equivalently of $\ket{G}$) is at most the $k$-rank-integrity of $G$. Now let $G'$ be a $k$-perturbation of $G$. By~\cite[Lemma 7.2]{campbell2026erdHos}, $G'$ is locally equivalent to a graph which is a rank-$2t$ perturbation of $G$. Since the maximum component size is invariant under local complementation (or, in the quantum picture, by one-qubit gates), it follows that the $2k$-rank-integrity is at most the $k$-ancilla integrity.
\end{proof}

Notice that this relationship holds for any graph parameter which is invariant under local complementation. More precisely, let $f(G)$ be a function that associates a value to a graph $G$, such that $f(G)$ is invariant by local complementation (that is, $f(G')=f(G)$ for any graph $G'$ locally equivalent to $G$). Then, the following generalization of \cref{lem:flipAncillaComparison} holds: for any graph $G$ and non-negative integer $k$,
\begin{align*}
    \min\{f(G') ~|~ G'  &\text{is a $(2k)$-rank-perturbation of $G$}\} \\
    & \leq \min\{f(G') ~|~ G' \text{ is a $k$-perturbation of $G$}\}\\
    & \leq \min\{f(G') ~|~ G' \text{is a $k$-rank-perturbation of $G$}\}
\end{align*}If instead we wish to maximize $f$, a similar relationship holds with the inequalities reversed.

\subsection{Comparison with Pauli persistency}

We compare ancilla-integrity to another natural notion of integrity for graph states, inspired by the Pauli persistency, defined below.

\begin{definition}[\cite{BRpersistency2001, Hein04}]
    The \emph{Pauli persistency} of a graph state $\ket{G}$ is the minimum number of one-qubit Pauli measurements on the qubits of $\ket{G}$ needed to disentangle $\ket{G}$.
\end{definition}

Equivalently, the \emph{Pauli persistency} of a graph state $\ket{G}$ is the smallest integer $k$ so that $G$ is locally equivalent to a graph $G'$ with a set $X$ of at most $k$ vertices so that $G'-X$ is an edgeless graph. The definition of Pauli persistency requires that the target quantum state is completely disentangled. We may ask a relaxed question: what is the smallest integer $k$ such that $k$ one-qubit Pauli measurements on the qubits of $\ket G$ deterministically maps $\ket{G}$ to a union of disconnected quantum states with at most $t$ qubits? In graph-theoretical terms, the question reads: what is the smallest integer $k$ such that $G$ is locally equivalent to a graph $G'$ with a set $X$ of at most $k$ vertices so that each component of $G'-X$ has at most $t$ vertices? The difference with $k$-ancilla-integrity is that the Pauli measurements are on the qubits of the graph state, and not on the ancilla qubits. 

We show below that allowing Pauli measurements only on the qubits of the graph state is much weaker: sometimes a Pauli measurement on a single ancilla qubit is enough to transform the graph state into a union of disconnected quantum states with a small number of qubits, but many Pauli measurements on the qubits of the systems are needed to do so.

\begin{proposition}
\label{prop:PauliVersusAncilla}
    For any non-negative integers $t$ and $k$, there exists a graph state with 1-ancilla-integrity $\leq k+1$, but for which any $k$ one-qubit Pauli measurements lead to a quantum state with a connected component containing at least $t$ qubits. 
\end{proposition}

\begin{proof}
    Let $m = \max(k+2,t)$. Let $G$ be a complete $m$-partite graph, where each vertex set in the partition contains exactly $k+1$ vertices; so $G$ is of order $ m (k+1)$. An analysis of the orbit of $G$ by local complementations is given in~\cite[Lemma~7.3]{campbell2026erdHos}. One can check that deleting $k$ vertices from any locally equivalent graph leads to a graph with a connected component containing at least $m$ vertices. 
    
    Conversely, adding a vertex $a$ to $G$ connected to every vertex, then locally complementing on $a$, then deleting $a$, maps $G$ to a disjoint union of $m$ complete graphs with $k+1$ vertices.
\end{proof}

\subsection{Comparison with order-integrity}

In this subsection we take a step back from the graph states interpretation in order to compare rank-integrity with another known parameter. 

Recall that the \emph{$k$-order-integrity} of a graph $G$ is the minimum, over all sets $X$ of at most $k$ vertices, of the maximum component size of $G-X$. We show the following quick lemma which connects these two parameters in one direction.

\begin{lemma}
\label{lem:denseAnalogEasy}
    For any graph $G$ and any non-negative integer $k<|V(G)|$,\begin{align*}
(2k)\text{-rank-integrity}(G) \leq k\text{-order-integrity}(G).
\end{align*}
\end{lemma}
\begin{proof}
    Let $X \subseteq V(G)$ be a set of size at most $k$. Let $G'$ be the subgraph of $G$ which is obtained by removing all edges incident to a vertex in $X$.
    Then, all the nonzero entries of the adjacency matrix of $G+G'$ are contained in at most $k$ rows and $k$ columns. 
    Thus, this matrix has rank at most $2k$. Moreover, the components of the graph $G'$ are exactly the components of $G-X$ and the vertices of $X$ (each of which is an isolated vertex).
    Thus, all components of $G'$ have size at most the $k$-order-integrity of $G$, as desired.
\end{proof}

For graphs that are even slightly sparse, it turns out that the connection also works in the other direction: rank integrity and order integrity agree up to a change in the parameter. 

\begin{lemma}
\label{lem:denseAnalog}
    For any non-negative integers $k$ and $t$ and any graph $G$ which does not contain a $K_{t,t}$-subgraph,\begin{align*}
    (2^{2k} \cdot 3t)\text{-order-integrity}(G) \leq k\text{-rank-integrity}(G).
    \end{align*}
\end{lemma}
\begin{proof}
    For this proof we use some definitions that are introduced in \Cref{subsec:XPalgorithmOverview}. Let $G^*$ be a (looped) graph of rank at most $k$ which minimizes the maximum component size of $G+G^*$.
    Let ${\mathcal{P}}^*$ be the type partition of $G^*$, and note that $|{\mathcal{P}}^*| \leq 2^k$ (since a rank-$k$ binary matrix has at most $2^k$ columns; see \cref{obs:rank-quotient}).

    For $P, P' \in {\mathcal{P}}^*$ (not necessarily distinct), let $H(P, P')$ be the graph on vertex set $P \cup P'$ with an edge between $v \in P$ and $v' \in P'$ if and only if $v$ and $v'$ are adjacent in ${G}^*$ and in different components of $G+{G}^*$. (So first we restrict our attention to what happens between the components of $G+{G}^*$, and then we include an edge if that adjacency is ``flipped''.) If $vv'$ is an edge of $H(P, P')$, then $v$ and $v'$ are adjacent in ${G}^*$ but not in $G + {G}^*$ since they belong to different components, so they are also adjacent in $G$. Thus, $H(P, P')$ is a subgraph of $G$ so contains no $K_{t, t}$-subgraph. Next we show how to find a small vertex cover of each graph $H(P, P')$; then we put all of these covers together to obtain a set $X$ so that every component of $G-X$ is small.

    \begin{claim}\label{cl:vertex-cover}
        For any $P, P' \in {\mathcal{P}}^*$, the graph $H(P, P')$ has a vertex cover of size at most~$3t$.
    \end{claim}

    \begin{subproof}{\cref{cl:vertex-cover}} If $|P| \leq 3t$ or $|P'| \leq 3t$ then we are done, so assume that $|P|, |P'| > 3t$.
        Let $V_1, \ldots, V_r$ denote the vertex sets of the components of $G + {G}^*$, and write $H=H(P, P')$.
        Since every pair of vertices in $P$ and every pair of vertices in $P'$ are twins in $G^*$, the sets $P$ and $P'$ are either complete or anticomplete in ${G}^*$. (See \cref{lem:parts-complete-anti}.) If they are anticomplete then $H$ is edgeless and we are done, so suppose that $P$ and $P'$ are complete in~${G}^*$.

        Then, for all $i \neq j \in [r]$, $P \cap V_i$ and $P' \cap V_j$ form a complete bipartite graph in $H$.
        Let $I \subseteq [r]$ be inclusion-wise minimal such that $|P \cap \bigcup_{i \in I}V_i| \geq t$.
        Then, $|P' \cap \bigcup_{j \notin I}V_j| < t$ since $H$ is $K_{t, t}$-free.
        Since $|P'| \geq 3t$, this implies $|P' \cap \bigcup_{i \in I}V_i| \geq 2t$.
        In turn, this implies that $|P \cap \bigcup_{j \notin I}V_j| < t$ since $H$ is $K_{t, t}$-free.

        Suppose first that $I = \{i\}$ for some $i \in [r]$. Then, $(P \cup P') \cap \bigcup_{j \notin I}V_j$ is a vertex cover of $H$ of size at most $2t$. So, suppose now that $|I| \geq 2$. By the minimality of $I$, we have $|P \cap \bigcup_{i \in I}V_i| \leq 2t$. Then, $|P| \leq 3t$ and we are done.
    \end{subproof}

    For all $P, P' \in {\mathcal{P}}^*$, let $X(P, P')$ be a vertex cover of $H(P, P')$ of size at most $3t$.
    Let $X$ be the union of all such sets and note that $|X| \leq 2^{2k} \cdot 3t$.
    To complete the proof, we now argue that every component of $G-X$ is contained in a component of $G + {G}^*$.
    For this, it suffices to show that any two adjacent vertices in $G-X$ lie in the same component of $G+{G}^*$.
    So, let $v, v' \in V(G) \setminus X$ be adjacent in $G$.
    Let $P$ (resp. $P'$) be the part of ${\mathcal{P}}^*$ that contains $v$ (resp. $v'$). Since $v$ and $v'$ are not in the chosen cover $X(P, P')$ of $H(P,P')$, either they are not adjacent in ${G}^*$ or they are in the same component of $G+{G}^*$. Since they are adjacent in $G$, if they are not adjacent in ${G}^*$ then they are in the same component of $G+{G}^*$. Thus, in any case they lie in the same component of $G+{G}^*$, which concludes the proof.
\end{proof}

\section{An \XP{} algorithm for \textsc{rank integrity}}
\label{sec:XPalgorithm}

\subsection{Basic tools}

Recall that the type partition of a looped graph $G$ is the partition where two vertices are in the same class if and only if they have the same neighborhood in $G$.

\begin{lemma}\label{lem:parts-complete-anti}
    Let $G$ be a looped graph and let $\mathcal{P}$ be the type partition of $G$. If $P_1, P_2 \in \mathcal{P}$ then $P_1$ and $P_2$ are either complete or anticomplete in $G$.
\end{lemma}

\begin{proof}
    Let $P_1, P_2 \in \mathcal{P}$ and let $u_1, u'_1 \in P_1$ and $u_2, u'_2 \in P_2$.
    Since $u_1 \sim u'_1$ and $u_2 \sim u'_2$, we have \[u_1u_2 \in E(G) \iff u'_1u_2 \in E(G) \iff u'_1u'_2 \in E(G).\]
    This shows that $P_1$ and $P_2$ are either complete or anticomplete in $G$.
\end{proof}

Recall that the quotient graph $G/\mathcal{P}$ is the looped graph on vertex set $\mathcal{P}$ where $P_1, P_2 \in \mathcal{P}$ are adjacent if and only if there is an edge between them in $G$ (in which case they are complete in $G$ by \cref{lem:parts-complete-anti}).
Note that $G/\mathcal{P}$ does not have twins by definition of $\mathcal{P}$; where we say that two vertices of a looped graph are twins if they have the same neighborhood.
Since $G/\mathcal{P}$ is an induced subgraph of $G$ and since $G$ can be obtained from $G/\mathcal{P}$ by substituting each looped vertex by a clique with a loop on each vertex and each non-looped vertex by a stable set, we have the following.

\begin{observation}\label{obs:rank-quotient}
    Let $G$ be a looped graph and let $\mathcal{P}$ be the type partition of $G$.
    Then, $G$ has rank at most $k$ if and only if $G/\mathcal{P}$ has rank at most $k$.
    Furthermore, if $G$ has rank at most $k$ then $|\mathcal{P}| \leq 2^k$.
\end{observation}

We now introduce a problem which will serve as a subroutine in several of our algorithms.
In the \textsc{weighted list balancing} problem, we are given as input a number $m$, items $v_1, \ldots, v_n$, each with a weight $w_i$ and a list $L_i \subseteq [m]$. 
The goal is to find an assignment $\alpha : [n] \to [m]$ such that $\alpha(i) \in L_i$ for every $i \in [n]$, that minimizes \[\max_{j \in [m]}\sum_{i : \alpha(i) = j}w_i.\]

This problem can be easily solved by dynamic programming, see for instance \cite[Exercise 8.18]{LS}. When the weights are given in unary, this gives an efficient algorithm for this problem.

\begin{lemma}\label{lem:list-balancing-easy}
    There is an algorithm solving \textsc{weighted list balancing} in time $\mathcal{O}(mnS^{m-1})$, where $S = \sum_{i \in [n]} w_i$.
\end{lemma}

\subsection{Notions for the algorithm}

Recall that in the \textsc{rank integrity} problem, we are given as input a graph $G$ and an integer $k$, and our goal is to find a rank-$k$ perturbation $G' = G + \widehat{G}$ of $G$ that minimizes the maximum component size of $G'$.

Note that there could be several rank-$k$ perturbations of $G$ whose maximum component size is minimum.
Among all such perturbations $G'$, we are going to focus on one where $\widehat{G}$ is ``as simple as possible'' in the following sense: for any two vertices $v, v'$ that are in the same component of $G'$ but do not have the same neighborhood in $\widehat{G}$, there is a vertex outside of their component in $G'$ on which they disagree.
This intuition is captured by the following definitions.

Let $G, \widehat{G}$ be two looped graphs on the same vertex set. 
Let $\widehat{\mathcal{P}}$ be the type partition of $\widehat{G}$.
The \emph{score} of $\widehat{G}$ with respect to $G$ is the sum over all components $K$ of $G + \widehat{G}$ of the number of parts of $\widehat{\mathcal{P}}$ that intersect $V(K)$.
We say that $\widehat{G}$ is \emph{$k$-optimal} for $G$ if $G+\widehat{G}$ is a rank-$k$ perturbation of $G$ that \begin{itemize}
    \item minimizes the maximum component size of $G+\widehat{G}$, 
    \item among all such $\widehat{G}$, maximizes the number of components of $G+\widehat{G}$, and
    \item among all such $\widehat{G}$, has minimum score with respect to $G$.
\end{itemize} 

\begin{lemma}\label{lem:distinguished-outside}
    Let $G$ be a graph and let $\widehat{G}$ be $k$-optimal for $G$. Let $\widehat{\mathcal{P}}$ be the type partition of $\widehat{G}$.
    Let $P \neq P' \in \widehat{\mathcal{P}}$ and let $u \in P, u' \in P'$ belong to the same component $K$ of $G+\widehat{G}$.
    Then, there exists $v \in V(G)\setminus K$ that is adjacent to exactly one of $u, u'$ in $\widehat{G}$.
\end{lemma}

\begin{proof}
    By contradiction, suppose that every $v \in V(G)$ that is adjacent to exactly one of $u, u'$ in $\widehat{G}$ belongs to $V(K)$.
    Consider the graph $\widetilde{G}$ on vertex set $V$ obtained from $\widehat{G}$ by setting the neighborhood of all vertices in $P' \cap V(K)$ to the neighborhood of $u$ in $\widehat{G}$ (equivalently, this amounts to moving them from $P'$ to $P$).
    Then, the only differences between $G+\widehat{G}$ and $G+\widetilde{G}$ are edges inside $V(K)$, so every component of $G+\widehat{G}$ is the union of components of $G+\widetilde{G}$.
    Note that the quotient graph $\widetilde{G}/\widetilde{\mathcal{P}}$ is an induced subgraph of the quotient graph of $\widehat{G}/\widehat{\mathcal{P}}$ (with the obvious notations), so $\widetilde{G}$ has rank at most $k$ by \cref{obs:rank-quotient}.
    Since the partition into components of $G+\widehat{G}$ minimizes the maximum component size, it follows that $G+\widehat{G}$ and $G+\widetilde{G}$ have the same maximum component size. 
    Since among all such graphs, $G+\widehat{G}$ has the maximum number of components, it follows that $G+\widehat{G}$ and $G+\widetilde{G}$ have the same components. 
    Then, $\widetilde{G}$ has smaller score than $\widehat{G}$ with respect to $G$, a contradiction to $\widehat{G}$ being $k$-optimal.
\end{proof}

In our algorithm for \textsc{rank integrity}, we will try to reconstruct a $k$-optimal $\widehat{G}$ for $G$ along with the partition into components of $G+\widehat{G}$.
To reconstruct the graph $\widehat{G}$, we will ``guess'' the quotient graph $\widehat{G}/\widehat{\mathcal{P}}$, where $\widehat{\mathcal{P}}$ is the type partition of $\widehat{G}$, and try to recover the partition $\widehat{\mathcal{P}}$.
To obtain $\widehat{\mathcal{P}}$ and the partition into components of $G+\widehat{G}$, we will ``guess'' some vertices from the intersection of some of the components of $G+\widehat{G}$ and the parts of $\widehat{\mathcal{P}}$ and try to reconstruct the partition into components around them. 
Thankfully, $\widehat{\mathcal{P}}$ only consists of a bounded number of parts.
However, there is no reason why the number of components of $G+\widehat{G}$ should be bounded, so we cannot ``guess'' vertices from all the components of $G+\widehat{G}$.
For this reason, we will carefully select a ``characteristic'' set of components of $G+\widehat{G}$ from which we will ``guess'' our vertices.
This is formalized in the next definitions.

Let $G, \widehat{G}$ be two looped graphs on the same vertex set and let $\widehat{\mathcal{P}} = \{P_1, \ldots, P_p\}$ be the type partition of $\widehat{G}$.
The \emph{signature} of a component $K$ of $G+\widehat{G}$ is the set $\{P \in \widehat{\mathcal{P}} : P \cap V(K) \neq \emptyset\}$.
A set $\mathcal{K}$ of components of $G+\widehat{G}$ is \emph{typical} if it is a maximum set of components of $G+\widehat{G}$ no four of which have the same signature. 
In other words, for each signature, $\mathcal{K}$ contains three components with that signature if at least three such components exist, and contains all of them otherwise. 
If $\widehat{G}$ has rank at most $k$, \cref{obs:rank-quotient} implies that $p = |\widehat{\mathcal{P}}| \leq 2^k$, so there are at most $2^{2^k}$ possible signatures for a component. 
Since no four components of $\mathcal{K}$ have the same signature, it follows that $|\mathcal{K}| \leq 3 \cdot 2^{2^k}$ for any typical set of components.
Let $\mathcal{K} \eqqcolon \{K_1, \ldots, K_t\}$ be a typical set of components of $G+\widehat{G}$.
Let $\mathcal{R} \coloneqq \{(i, j) \in [t] \times [p] : V(K_i) \cap P_j \neq \emptyset\}$.
Note that for every $j \in [p]$, there exists $i \in [t]$ such that $(i, j) \in \mathcal{R}$.
For each $(i, j) \in \mathcal{R}$, pick an arbitrary \emph{representative} $v_{i, j} \in V(K_i) \cap P_j$, and let $R$ denote the family $(v_{i, j})_{(i, j) \in \mathcal{R}}$.
We say that the tuple $(t, p, \mathcal{R}, R, \widehat{G}/\widehat{\mathcal{P}})$ is a \emph{typical system of representatives} for $(G, \widehat{G})$.

\subsection{The algorithm}

We now turn to the main result of this section, which is restated below for convenience.

\thmAlgorithm*

We provide a high-level description of the algorithm for \cref{thm:algo-XP} in \cref{alg:pseudocode}. Formally, each ``Guess'' step corresponds to a brute-force search over all possibilities, and ``reject the guess'' means that the current branch is discarded.
The algorithm returns the optimal value over all non-rejected branches.

\begin{algorithm}
    \caption{Roadmap of the algorithm for \cref{thm:algo-XP}}\label{alg:pseudocode}
    \label{algo-XP}
    \begin{algorithmic}[1]
    \Require Graph $G$ and integer $k$
    \State Guess a possible typical system of representatives $(t, p, \mathcal{R}, R, H)$
    \State If $R$ is not consistent with $H$: reject the guess \label{line:consistent}
    \State Compute the compatibility function $f$
    \State If $f$ is not faithful or not functional: reject the guess \label{line:faithful-functional}
    \State Compute the strong conflict classes $\mathcal{U}$ and the set $X$ of crossing tuples
    \State Guess a crossing solver $\phi$ such that for every $x \in X$, we have $\phi(x) \in W_x \cup W'_x \cup \{\bot\}$
    \State Compute the updated compatibility function $\hat{f}$ and the eligibility map $\tau$ on $\mathcal{U}$
    \State Compute a suitable partition $\mathcal{C}$ of minimum width using \cref{lem:list-balancing-easy}
    \State If no such partition exists: reject the guess
    \State Compute $\psi : V(G) \to V(H)$ and build $\widehat{G}$
    \State If the partition into components of $G+\widehat{G}$ does not refine $\mathcal{C}$: reject the guess
    \State Return the width of $\mathcal{C}$
    \end{algorithmic}
\end{algorithm}

\begin{proof}
    Suppose that we are given an instance $(G, k)$ of \textsc{rank integrity}.
    For the rest of this proof, we fix a looped graph $G^*$ that is $k$-optimal for $G$ and a typical system of representatives $(t^*, p^*, \mathcal{R}^*, R^*, G^*/\mathcal{P}^*)$ for $(G, G^*)$. 
    We denote by $\mathcal{K}^* = \{K_1, \ldots, K_{t^*}\}$ the typical set of components of $G+G^*$ that corresponds to this typical system of representatives and by $K_1, \ldots, K_{c^*}$ the components of $G+G^*$.
    Note that $t^*$ is the number of components in the typical set of components $\mathcal{K}^*$, while $c^*$ is the number of components of $G+G^*$.
    We also write $\mathcal{P}^* = \{P_1, \ldots, P_{p^*}\}$, where $p^*$ is the number of parts of $\mathcal{P}^*$.
    Our goal is to reconstruct the partition into components of $G+G^*$ along with the graph $G^*$. 
    For this, we want to start from the typical system of representatives $(t^*, p^*, \mathcal{R}^*, R^*, G^*/\mathcal{P}^*)$.
    Thus, the first step of the algorithm is to guess: \begin{itemize}
        \item an integer $t \leq 3 \cdot 2^{2^k}$,
        \item an integer $p \leq 2^k$,
        \item a set $\mathcal{R} \subseteq [t] \times [p]$ such that for every $j \in [p]$, there exists $i \in [t]$ such that $(i, j) \in \mathcal{R}$,
        \item a family $R = (v_{i, j})_{(i, j) \in \mathcal{R}}$ of vertices of $G$, and
        \item a looped graph $H$ on vertex set $[p]$ of rank at most $k$ that does not contain twins.
    \end{itemize}     
    Note that $(t^*, p^*, \mathcal{R}^*, R^*, G^*/\mathcal{P}^*)$ indeed satisfies all these conditions.
    
    Throughout the proof, we will establish properties regarding $(t^*, p^*, \mathcal{R}^*, R^*, G^*/\mathcal{P}^*)$ for $(G, G^*)$. The algorithm will verify that the guess $(t, p, \mathcal{R}, R, H)$ satisfies these properties, and abort if it is not the case.

    \paragraph*{Consistency.} We start with a first property.
    We say that $R$ is \emph{consistent} with $H$ if for all $v_{i, j}, v_{i', j'} \in R$ with $i \neq i'$, we have $v_{i, j}v_{i', j'} \in E(G) \iff jj' \in E(H)$.
    
    \begin{claim}\label{cl:consistent}
        $R^*$ is consistent with $G^*/\mathcal{P}^*$.
    \end{claim}

    \begin{subproof}{\cref{cl:consistent}}
        Let $v_{i, j}, v_{i', j'} \in R^*$ with $i \neq i'$. 
        Since $v_{i, j}$ and $v_{i', j'}$ belong to different components of $G+G^*$, they are not adjacent in $G+G^*$. 
        Therefore, $v_{i, j}v_{i', j'} \in E(G) \iff v_{i, j}v_{i', j'} \in E(G^*) \iff jj' \in E(G^*/\mathcal{P}^*)$.
    \end{subproof}

    The algorithm verifies that $R$ is consistent with $H$, and aborts if it is not the case.
    We then start trying to construct the type partition $\mathcal{P}^*$ of $G^*$ and the partition into components of $G+G^*$.
    We begin with an observation.
    
    Let $v \in V(G)$ and $v_{i, j} \in R^*$ and suppose that $v$ and $v_{i, j}$ belong to the same component $K_i$ of $G+G^*$ and to the same part $P_j$ of the type partition $\mathcal{P}^*$ of $G^*$. 
    Then, $v$ and $v_{i, j}$ have the same neighborhood in $G^*$, and they are not adjacent in $G+G^*$ to any vertex $v_{i', j'}$ for $i' \neq i$. 
    Thus, $v$ and $v_{i, j}$ have the same neighborhood in $G$ on all the vertices $v_{i', j'} \in R^*$ with $i' \neq i$.
    
    \paragraph*{Compatibility.} With this idea in mind, we say that a vertex $v \in V(G)$ is \emph{compatible} with a vertex $v_{i, j} \in R$ if $v$ and $v_{i, j}$ have the same neighborhood in $G$ on all the vertices $v_{i', j'} \in R$ with $i' \neq i$.
    We then define the \emph{compatibility function} $f : V(G) \to 2^R$ as follows: \begin{itemize}
        \item If $v \in R$ then $f(v) = \{v\}$, 
        \item If $v \notin R$ then $f(v) = \{v_{i, j} \in R : \text{$v$ is compatible with $v_{i, j}$}\}$.
    \end{itemize}
    We denote by $f^*$ the compatibility function corresponding to $(t^*, p^*, \mathcal{R}^*, R^*, G^*/\mathcal{P}^*)$.
    Intuitively, for every vertex $u \in V(G)$, we should think of $f(u)$ as describing the parts of $\mathcal{P}^*$ and the components of $G+G^*$ (more precisely, those among $K_1, \ldots, K_{t^*}$) to which $u$ can belong.

    \paragraph*{Faithfulness.}
    We say that $f$ is \emph{faithful} if $|f(u)| \geq 1$ for every $u \in V(G)$.
    This means that for every $u \in V(G)$, there is some component $K_i \in \mathcal{K}^*$ with and some part $P_j$ of $\mathcal{P}^*$ where $u$ can go.

    \begin{claim}\label{cl:f*-faithful}
        Let $u \in V(G)$, let $i_u \in [c^*]$ such that $u \in V(K_{i_u})$ and let $j_u \in [p^*]$ such that $u \in P_{j_u}$.
        If $i_u \in [t^*]$ then $u$ is compatible with $v_{i_u, j_u}$.
        If $i_u \notin [t^*]$ then $u$ is compatible with $v_{i, j_u}$ for all $i \in [t^*]$ such that $(i, j_u) \in \mathcal{R}^*$.
        In particular, $f^*$ is faithful.
    \end{claim}

    \begin{subproof}{\cref{cl:f*-faithful}}
        Suppose first that $i_u \in [t^*]$, so there is a vertex $v_{i_u, j_u} \in R^*$.
        In this case, we already argued that $v$ and $v_{i_u, j_u}$ are compatible.
    
        Suppose now that $i_u \notin [t^*]$ and let $i \in [t^*]$ such that $(i, j_u) \in \mathcal{R}^*$.
        We show that $v$ is compatible with $v_{i, j_u}$.
        Let $i' \neq i \in [t^*]$ and $j' \in [p^*]$ such that $v_{i', j'} \in R^*$.
        Observe that $K_{i_u} \neq K_{i'}$ since $i_u \notin [t^*]$.
        Thus, both $v$ and $v_{i, j_u}$ are not adjacent to $v_{i', j'}$ in $G+G^*$.
        Since $v, v_{i, j_u} \in P_{j_u}$, $v$ and $v_{i, j_u}$ have the same neighborhood in $G^*$.
        Thus, $v$ and $v_{i, j_u}$ have the same adjacency in $G$ to $v_{i', j'}$.
        This proves that $v$ and $v_{i, j_u}$ are compatible.

        The fact that $f^*$ is faithful follows immediately, since there exists $i \in [t^*]$ such that $(i, j_u) \in \mathcal{R}^*$.
    \end{subproof}

    \paragraph*{Functionality.}
    We say that $f$ is \emph{functional} if for every $u \in V(G)$, for every $i \in [t]$, there is at most one $j \in [p]$ such that $v_{i, j} \in f(u)$.
    Informally, this means that $v$ cannot be compatible with two representatives that live in the same component of $G^*$.

    \begin{claim}\label{f*-functional}
        $f^*$ is functional.
    \end{claim}

    \begin{subproof}{\cref{f*-functional}}
        Let $u \in V(G)$, $i \in [t^*]$ and $j \neq j' \in [p^*]$ such that $v_{i, j}, v_{i, j'} \in R^*$. 
        By \cref{lem:distinguished-outside}, there exists $w \in V(G)$ that does not belong to $K_i$ and that is adjacent to exactly one of $v_{i, j}, v_{i, j'}$ in $G^*$.
        Let $j_0 \in [p^*]$ such that $w \in P_{j_0}$.
        Since $\mathcal{K}^*$ is a typical set of components of $G+G^*$, there exists $i_0 \neq i \in [t^*]$ such that $(i_0, j_0) \in \mathcal{R}^*$, so there is a vertex $v_{i_0, j_0} \in R^*$.
        By \cref{lem:parts-complete-anti}, it follows that $v_{i_0, j_0}$ is adjacent to exactly one of $v_{i, j}, v_{i, j'}$ in $G^*$.
        Since $i_0 \neq i$, $v_{i_0, j_0}$ is adjacent to neither $v_{i, j}$ nor $v_{i, j'}$ in $G+G^*$.
        Thus, $v_{i_0, j_0}$ is adjacent to exactly one of $v_{i, j}, v_{i, j'}$ in $G$. Therefore, $u$ cannot be compatible with both $v_{i, j}$ and $v_{i, j'}$.
    \end{subproof}

    The algorithm verifies that $f$ is faithful and functional, and aborts if it is not the case.
    From there, we would like to be able to decide which part of $\mathcal{P}^*$ each $u \in V(G)$ belongs to.
    If $u$ is compatible with several vertices that belong to the same part $P_j$ of $\mathcal{P}^*$, this is a good indication that maybe we should try to put $u$ in the part $P_j$.
    Formally, a vertex $u \in V(G)$ is \emph{flagged} if there exists $j \in [p]$ such that $|\{i \in [t] : v_{i, j} \in f(u)\}| \geq 2$.
    This terminology is justified by the following observation -- which holds even if the guess is not $(t^*, p^*, \mathcal{R}^*, R^*, G^*/\mathcal{P}^*)$.

    \begin{claim}\label{cl:flagged-unique}
        Let $u \in V(G)$ be a flagged vertex. 
        There exists a unique $j \in [p]$ such that $|\{i \in [t] : v_{i, j} \in f(u)\}| \geq 2$.
    \end{claim}
    
    \begin{subproof}{\cref{cl:flagged-unique}}
        The existence of $j$ follows from the definition of flagged. 
        We now prove the uniqueness.
        By contradiction, suppose that there exist $j \neq j' \in [p]$, $i_1 \neq i_2 \in [t]$, and $i'_1 \neq i'_2 \in [t]$ such that $v_{i_1, j}, v_{i_2, j}, v_{i'_1, j'}, v_{i'_2, j'} \in f(u)$.
        Since $H$ is twin-free, there exists $j_0 \in [p]$ such that exactly one of $jj_0, j'j_0$ is an edge of $H$.
        By construction, there exists $i_0 \in [t]$ such that $(i_0, j_0) \in \mathcal{R}$, so there is a vertex $v_{i_0, j_0} \in R$.
        Without loss of generality, we may assume that $i_1 \neq i_0 \neq i'_1$.
        Then, since $R$ is consistent with $H$ (otherwise the algorithm would have aborted in Line~\ref{line:consistent}) and $u$ is compatible with both $v_{i_1, j}$ and $v_{i'_1, j'}$, we have \begin{align*}
            jj_0 \in E(H) &\iff v_{i_1, j}v_{i_0, j_0} \in E(G)  \\
                &\iff uv_{i_0, j_0} \in E(G) \\
                &\iff v_{i'_1, j'}v_{i_0, j_0} \in E(G) \\
                &\iff j'j_0 \in E(H),
        \end{align*} a contradiction.
    \end{subproof}

    \paragraph*{Strong conficts.}
    The information stored by $f$ is only based on the family $R$ that we guessed, and it can be inconsistent between different vertices of $G$.
    Say that two vertices $u, u' \in V(G)$ are \emph{strongly conflicting} if for every $v_{i, j} \in f(u)$ and every $v_{i', j'} \in f(u')$ such that $i \neq i'$, the adjacency in $G$ between $v_{i, j}$ and $v_{i', j'}$ is different from that between $u$ and $u'$.
    This means that if we want to put $u$ and $u'$ in the components and parts described by $f$, they must be placed in the same component.
    
    The reflexive transitive closure of the ``strongly conflicting'' relation is an equivalence relation.
    Let $\mathcal{U}$ be the partition of $V(G)$ into equivalence classes under this relation.
    We denote by $\mathcal{U}^*$ the partition corresponding to the function $f^*$.
    As we argued above, the only way to reconcile two strongly conflicting vertices should be to put them in the same component of the perturbation. This is indeed the case in $G^*$.

    \begin{claim}\label{cl:conflicting-same-cc}
        Every $U \in \mathcal{U}^*$ is contained in a component of $G+G^*$.
    \end{claim}    

    \begin{subproof}{\cref{cl:conflicting-same-cc}}
        It suffices to prove that if $u, w \in V(G)$ are strongly conflicting for $f^*$, they belong to the same component of $G+G^*$.
        Let $u, w \in V(G)$ be strongly conflicting for $f^*$.
        Let $K_{i_u}$ (resp. $K_{i_{w}}$) be the component of $G+G^*$ that contains $u$ (resp. $w$) and let $P_{j_u}$ (resp. $P_{j_{w}}$) be the part of $\mathcal{P}^*$ that contains $u$ (resp. $w$).
        By contradiction, suppose that $i_u \neq i_{w}$.

        We first show that there exist vertices $u', w' \in V(G)$ such that: \begin{itemize}
            \item all vertices among $u, u', w, w'$ belong to different components of $G+G^*$, except maybe $\{u, u'\}$, and $\{w, w'\}$,
            \item $u' \in P_{j_u}$ and $w' \in P_{j_w}$, and
            \item the number of edges between $\{u, u'\}$ and $\{w, w'\}$ in $G$ is nonzero modulo $4$.
        \end{itemize}

        Consider first the case where $i_u, i_{w} \in [t^*]$.
        Then, there are vertices $v_{i_u, j_u}, v_{i_{w}, j_{w}} \in R^*$, and $u$ is compatible with $v_{i_u, j_u}$ and $w$ with $v_{i_{w}, j_{w}}$ by \cref{cl:f*-faithful}.
        Since $u, w$ are strongly conflicting, exactly one of $uw$ and $v_{i_u, j_u}v_{i_{w}, j_{w}}$ is an edge in $G$, so the number of edges between $\{u, v_{i_u, j_u}\}$ and $\{w, v_{i_{w}, j_{w}}\}$ in $G$ is nonzero modulo 4.
        In this case, set $u' = v_{i_u, j_u}$ and $w' = v_{i_w, j_w}$.
        
        Suppose now that $i_u \notin [t^*]$ and $i_w \in [t^*]$.
        By definition of $\mathcal{K}^*$, there exists $i'_u \neq i_w \in [t^*]$ such that $v_{i'_u, j_u} \in R^*$.
        Then, $u$ is compatible with $v_{i'_u, j_u}$ and $w$ with $v_{i_w, j_w}$ by \cref{cl:f*-faithful}.
        Since $u, w$ are strongly conflicting, exactly one of $uw$ and $v_{i'_u, j_u}v_{i_w, j_w}$ is an edge in $G$, so the number of edges between $\{u, v_{i'_u, j_u}\}$ and $\{w, v_{i_w, j_w}\}$ in $G$ is nonzero modulo 4.
        In this case, set $u' = v_{i'_u, j_u}$ and $w' = v_{i_w, j_w}$.
        The case where $i_u \in [t^*]$ and $i_w \notin [t^*]$ is symmetric.
        
        Suppose finally that $i_u, i_w \notin [t^*]$.
        By definition of $\mathcal{K}^*$, there exist $i'_u \in [t^*]$ such that $v_{i'_u, j_u} \in R^*$, and $i'_w \neq i'_u \in [t^*]$ such that $v_{i'_w, j_w} \in R^*$.
        Then, $u$ is compatible with $v_{i'_u, j_u}$ and $w$ with $v_{i'_w, j_w}$ by \cref{cl:f*-faithful}.
        Since $u, w$ are strongly conflicting, exactly one of $uw$ and $v_{i'_u, j_u}v_{i'_w, j_w}$ is an edge in $G$, so the number of edges between $\{u, v_{i'_u, j_u}\}$ and $\{w, v_{i'_w, j_w}\}$ in $G$ is nonzero modulo 4.
        In this case, set $u' = v_{i'_u, j_u}$ and $w' = v_{i'_w, j_w}$.

        This proves the existence of $u', w'$ with the desired properties.
        Since $u, u' \in P_{j_u}$ and $w, w' \in P_{j_{w}}$, the number of edges between $\{u, u'\}$ and $\{w, w'\}$ in $G^*$ is zero modulo 4.
        Thus, the number of edges between $\{u, u'\}$ and $\{w, w'\}$ in $G+G^*$ is nonzero modulo 4. 
        This contradicts that all these vertices belong to different components of $G+G^*$, except maybe $u$ and $u'$, and $w$ and $w'$.
        This concludes the proof that $i_u = i_w$, so $u, w$ belong to the same component of $G+G^*$.
    \end{subproof}

    \paragraph*{Weak conflicts.}
    There can also be inconsistencies between vertices that are not strongly conflicting.
    Say that two vertices $u, u' \in V(G)$ are \emph{conflicting} if there exist $v_{i, j} \in f(u)$ and $v_{i', j'} \in f(u')$ with $i \neq i'$ such that the adjacency in $G$ between $v_{i, j}$ and $v_{i', j'}$ is different from that between $u$ and $u'$.
    They are \emph{weakly conflicting} if they are conflicting but not strongly conflicting.
    We show that weakly conflicting vertices have a very specific behavior, even if the guess is not $(t^*, p^*, \mathcal{R}^*, R^*, G^*/\mathcal{P}^*)$.

    \begin{claim}\label{cl:setup-weakly-conflicting}
        Let $u, u' \in V(G)$ be weakly conflicting.
        Then, there exist $i_1 \neq i_2 \in [t]$ and $j_1, j_2, j'_1, j'_2 \in [p]$ such that $f(u) = \{v_{i_1, j_1}, v_{i_2, j_2}\}$ and $f(u') = \{v_{i_1, j'_1}, v_{i_2, j'_2}\}$.
        Furthermore, exactly one of $v_{i_1, j_1}v_{i_2, j'_2}$ and $v_{i_1, j'_1}v_{i_2, j_2}$ is an edge in $G$.
    \end{claim}
    
    \begin{subproof}{\cref{cl:setup-weakly-conflicting}}
        Since $f$ is faithful (otherwise the algorithm would have aborted in Line~\ref{line:faithful-functional}), we have $|f(u)|, |f(u')| \geq 1$.
        By contradiction, suppose that $|f(u)| = 1$ and let $f(u) \coloneqq \{v_{i, j}\}$.
        By definition, every $v_{i', j'} \in f(u')$ with $i' \neq i$ has the same adjacency as $u'$ on $v_{i, j}$ in $G$. 
        Thus, either $u$ and $u'$ are strongly conflicting or they are not conflicting, which is impossible by assumption.
        Thus, $|f(u)| \geq 2$, and by symmetry $|f(u')| \geq 2$.
    
        By contradiction, suppose that $|f(u)| \geq 3$ and let $v_{i_1, j_1}, v_{i_2, j_2}, v_{i_3, j_3} \in f(u)$. 
        Since $f$ is functional (otherwise the algorithm would have aborted in Line~\ref{line:faithful-functional}), $i_1, i_2$ and $i_3$ are distinct.
        Let $v_{i'_1, j'_1}, v_{i'_2, j'_2} \in f(u')$. 
        Again, $i'_1$ and $i'_2$ are distinct. 
        Without loss of generality, we can assume that $i_1, i_2, i_3, i'_1, i'_2$ are pairwise distinct, except maybe $i_1, i'_1$, and $i_2, i'_2$.
        By symmetry, suppose that $u'$ is adjacent to $v_{i_3, j_3}$ in $G$.
        Since $u'$ is compatible with $v_{i'_1, j'_1}$ and $v_{i'_2, j'_2}$, they are both adjacent to $v_{i_3, j_3}$ in $G$.
        Since $u$ is compatible with $v_{i_3, j_3}$, $u$ is adjacent to both $v_{i'_1, j'_1}$ and $v_{i'_2, j'_2}$ in $G$.
        Since $u$ is compatible with $v_{i_1, j_1}$ and $v_{i_2, j_2}$, they are both adjacent to $v_{i'_1, j'_1}$ and $v_{i'_2, j'_2}$ in $G$, except maybe $v_{i_1, j_1}$ and $v_{i'_1, j'_1}$ if $i_1 = i'_1$, and $v_{i_2, j_2}$ and $v_{i'_2, j'_2}$ if $i_2 = i'_2$.
        Since $v_{i_1, j_1}, v_{i_2, j_2}, v_{i_3, j_3} \in f(u)$ and $v_{i'_1, j'_1}, v_{i'_2, j'_2} \in f(u')$ were arbitrary, it follows that either $u, u'$ are strongly conflicting or they are not conflicting. 
        Both these outcomes are impossible, so $|f(u)| = 2$ and symmetrically $|f(u')| = 2$.
    
        Write $f(u) = \{v_{i_1, j_1}, v_{i_2, j_2}\}$ and $f(u') = \{v_{i'_1, j'_1}, v_{i'_2, j'_2}\}$. 
        Since $f$ is functional, $i_1$ and $i_2$ are distinct and so are $i'_1$ and $i'_2$.
        By contradiction, suppose that $\{i_1, i_2\} \neq \{i'_1, i'_2\}$.
        Without loss of generality, we can assume that $i_1, i_2, i'_1, i'_2$ are pairwise distinct, except maybe $i_2, i'_2$.
        By symmetry, suppose that $u'$ is adjacent to $v_{i_1, j_1}$ in $G$.
        Since $u'$ is compatible with $v_{i'_1, j'_1}$ and $v_{i'_2, j'_2}$, they are both adjacent to $v_{i_1, j_1}$ in $G$.
        Since $u$ is compatible with $v_{i_1, j_1}$, $u$ is adjacent to both $v_{i'_1, j'_1}$ and $v_{i'_2, j'_2}$ in $G$.
        Since $u$ is compatible with $v_{i_2, j_2}$, $v_{i_2, j_2}$ is adjacent to $v_{i'_1, j'_1}$ in $G$, and to $v_{i'_2, j'_2}$ in $G$ unless maybe if $i_2 = i'_2$.
        Then, it follows that either $u, u'$ are strongly conflicting or they are not conflicting. 
        Both these outcomes are impossible, so $\{i_1, i_2\} = \{i'_1, i'_2\}$.
        This proves the first part of the statement.
        Since $u$ and $u'$ are weakly conflicting, it follows immediately that exactly one of $v_{i_1, j_1}v_{i_2, j'_2}$ and $v_{i_1, j'_1}v_{i_2, j_2}$ is an edge in $G$.
    \end{subproof}

    \paragraph*{Crossings.}
    \cref{cl:setup-weakly-conflicting} motivates the following definition.
    Let $i_1 \neq i_2 \in [t]$ and $j_1, j_2, j'_1, j'_2 \in [p]$ such that $(i_1, j_1), (i_1, j'_1), (i_2, j_2), (i_2, j'_2) \in \mathcal{R}$.
    We say that $(i_1, i_2, j_1, j'_1, j_2, j'_2)$ is \emph{crossing} if $v_{i_1, j_1}v_{i_2, j'_2} \in E(G)$ and $v_{i_1, j'_1}v_{i_2, j_2} \notin E(G)$.
    Let $X \subseteq [t]^2 \times [p]^4$ denote the set of all such tuples that are crossing.
    Let $X^* \subseteq [t^*]^2 \times [p^*]^4$ be the corresponding set for $(t^*, p^*, \mathcal{R}^*, R^*, G^*/\mathcal{P}^*)$.

    Thus, \cref{cl:setup-weakly-conflicting} asserts that if $u, u' \in V(G)$ are weakly conflicting, there exists $x = (i_1, i_2, j_1, j_2, j'_1, j'_2) \in X$ such that $f(u) = \{v_{i_1, j_1}, v_{i_2, j_2}\}$ and $f(u') = \{v_{i_1, j'_1}, v_{i_2, j'_2}\}$, say.
    Suppose that we know that $u, u' \in V(K_{i_1} \cup K_{i_2})$ and that if $u \in V(K_{i_1})$ then $u \in P_{j_1}$ and so on.
    In that case, if $uu' \in E(G)$, we have $u' \in K_{i_1} \Rightarrow u \in K_{i_1}$, and if $uu' \notin E(G)$, we have $u \in K_{i_1} \Rightarrow u' \in K_{i_1}$.
    We represent these constraints as follows.

    Let $x = (i_1, i_2, j_1, j_2, j'_1, j'_2) \in X$.
    Let $W_x = \{u \in V(G) : f(u) = \{v_{i_1, j_1}, v_{i_2, j_2}\}\}$ and $W'_x = \{u \in V(G) : f(u) = \{v_{i_1, j'_1}, v_{i_2, j'_2}\}\}$.
    The \emph{crossing digraph for $x$} is the directed graph $D_x$ on vertex set $W_x \cup W'_x$ with an arc from $u \in W_x$ to $u' \in W'_x$ if and only if $uu' \notin E(G)$ and an arc from $u' \in W'_x$ to $u \in W_x$ if and only if $uu' \in E(G)$.
    We show how this digraph can be used to solve the weak conflicts between $W_x$ and $W'_x$ for $(G, G^*)$.

    \begin{claim}\label{cl:exists-crossing-solver}
        For every $x = (i_1, i_2, j_1, j_2, j'_1, j'_2) \in X^*$, one of the following holds in $G+G^*$: \begin{itemize}
            \item $W_x \cup W'_x \subseteq V(K_{i_2})$.
            \item There exists $u_0 \in W_x \cup W'_x$ such that every $u \in W_x \cup W'_x$ that can be reached from $u_0$ in $D_x$ satisfies $u \in V(K_{i_1})$, and every $u \in W_x \cup W'_x$ that cannot be reached from $u_0$ in $D_x$ but that can reach $u_0$ in $D_x$ satisfies $u \in V(K_{i_2})$.
        \end{itemize}
    \end{claim}
    
    \begin{subproof}{\cref{cl:exists-crossing-solver}}
        Let $\mathcal{S} = \{S_1, \ldots, S_s\}$ be the set of strongly connected components of $D_x$.
        Let $D'_x$ be the directed acyclic graph on vertex set $\mathcal{S}$ obtained by contracting each $S_i$ to a single vertex.
        It follows from \cref{cl:f*-faithful} that $W_x \cup W'_x \subseteq V(K_{i_1}) \cup V(K_{i_2})$.
        \cref{cl:f*-faithful} also implies that for every $u \in W_x$, if $u \in V(K_{i_1})$ then $u \in P_{j_1}$, and so on.
        Thus, as we argued before, an arc from $u$ to $v$ in $D_x$ means that $u \in V(K_{i_1})$ implies $v \in V(K_{i_1})$.        
        Thus, for every $a \in [s]$, either $S_a \subseteq V(K_{i_1})$ or $S_a \subseteq V(K_{i_2})$.
        Without loss of generality, we can assume that $S_1 \prec S_2 \prec \ldots \prec S_s$ is a topological ordering of $D'_x$.
        
        If $W_x \cup W'_x \subseteq V(K_{i_2})$, we are done, so suppose that $W_x \cup W'_x \not\subseteq V(K_{i_2})$.
        Let $a \in [s]$ be minimum such that $S_a \subseteq V(K_{i_1})$, which is well-defined since $W_x \cup W'_x \not\subseteq V(K_{i_2})$.
        Let $u_0$ be an arbitrary vertex of $S_a$.
        Let $u \in V(D_x)$ be reachable from $u_0$ in $D_x$.
        Since $u_0 \in V(K_{i_1})$ and since there exists a directed path from $u_0$ to $u$ in $D_x$, it follows that $u \in V(K_{i_1})$.
        Let $u \in V(D)$ that cannot be reached from $u_0$ but can reach $u_0$ in $D_x$.
        Let $b \in [s]$ such that $u \in S_b$.
        Since $u$ can reach $u_0 \in S_a$, we have $b \leq a$.
        Since $S_a$ is strongly connected in $D_x$ and $u_0$ cannot reach $u$ in $D_x$, we have $b \neq a$, so $b < a$.
        Thus, by minimality of $a$, we have $S_b \not \subseteq V(K_{i_1})$, so $S_b \subseteq V(K_{i_2})$ so $u \in V(K_{i_2})$.
    \end{subproof}

    \paragraph*{Crossing solvers.}
    A function $\phi : X^* \to V(G) \cup \{\bot\}$ is a \emph{crossing solver} for $(t^*, p^*, \mathcal{R}^*, R^*, G^*/\mathcal{P}^*)$ if the following holds for every $x = (i_1, i_2, j_1, j_2, j'_1, j'_2) \in X^*$: \begin{itemize}
        \item If $\phi(x) = \bot$ then $W_x \cup W'_x \subseteq V(K_{i_2})$.
        \item Otherwise, $\phi(x) \in W_x \cup W'_x$,  every $u \in W_x \cup W'_x$ that can be reached from $\phi(x)$ in $D_x$ satisfies $u \in V(K_{i_1})$, and every $u \in W_x \cup W'_x$ that cannot be reached from $\phi(x)$ in $D_x$ but that can reach $\phi(x)$ in $D_x$ satisfies $u \in V(K_{i_2})$.
    \end{itemize}
    \cref{cl:exists-crossing-solver} exactly says that there exists a crossing solver for $(t^*, p^*, \mathcal{R}^*, R^*, G^*/\mathcal{P}^*)$. Let $\phi^*$ be such a crossing solver.

    In the algorithm, we now guess a function $\phi : X \to V(G) \cup \bot$ such that for every $x \in X$, we have $\phi(x) \in W_x \cup W'_x \cup \{\bot\}$.

    \paragraph*{Updated compatibility.}
    Given \cref{cl:exists-crossing-solver}, it is natural to update the compatibility function $f$ to an \emph{updated compatibility function} $\widehat{f} : V(G) \to 2^R$ as follows.
    For every $u \in V(G)$, $\widehat{f}(u)$ is the maximal subset of $f(u)$ that satisfies the following conditions. 
    For every crossing $x = (i_1, i_2, \cdot, \cdot, \cdot, \cdot) \in X$ such that $u \in W_x \cup W'_x$, we have: \begin{itemize}
        \item If $\phi(x) = \bot$ then $\widehat{f}(u) \subseteq \{v_{i_2, j_2} : j_2 \in [p], (i_2, j_2) \in \mathcal{R}\}$.
        \item Otherwise, let $D_x$ be the crossing digraph for $x$. If $\phi(x)$ can reach $u$ in $D_x$ then $\widehat{f}(u) \subseteq \{v_{i_1, j_1} : j_1 \in [p], (i_1, j_1) \in \mathcal{R}\}$. If $\phi(x)$ cannot reach $u$ but $u$ can reach $\phi(x)$ in $D_x$ then $\widehat{f}(u) \subseteq \{v_{i_2, j_2} : j_2 \in [p], (i_2, j_2) \in \mathcal{R}\}$.
    \end{itemize}
    We then define $\tau : \mathcal{U} \to [t]$ as follows: for every $U \in \mathcal{U}$, $\tau(U) = \{i \in [t] : \forall u \in U, \exists j \in [p] : v_{i, j} \in \widehat{f}(u)\}$.
    Intuitively, $\widehat{f}$ stores where each vertex can go with respect to $\mathcal{P}^*$ and the partition into components of $G+G^*$. Given \cref{cl:conflicting-same-cc}, for every $U \in \mathcal{U}$, we want to put all vertices of $U$ in the same component of $G+G^*$. Then, $\tau(U)$ represents the set of all eligible components (among those from $\mathcal{K}^*$).
    Let $\widehat{f^*}$ be the updated compatibility function obtained from $f^*$ and $\phi^*$, and let $\tau^* : \mathcal{U}^* \to [t^*]$ be the corresponding function.
    It will follow from \cref{cl:updated-compatibility-info} that $\tau^*(U) \neq \emptyset$ for every $U \in \mathcal{U}^*$.

    The next lemma gives a sufficient condition for a part $U \in \mathcal{U}^*$ to be contained in one of the components that were selected in $\mathcal{K}^*$. 
    It turns out that this condition is also necessary, although we do not state it since we will not use it (explicitly).
    
    \begin{claim}\label{cl:updated-compatibility-info}
        If $U \in \mathcal{U}^*$ is such that $|\tau^*(U)| \leq 2$ or some $u \in U$ is not flagged then there exists $i \in \tau^*(U)$ such that $U \subseteq V(K_i)$.
    \end{claim}
    
    \begin{subproof}{\cref{cl:updated-compatibility-info}}
        Consider such a $U \in \mathcal{U}^*$. 
        By \cref{cl:conflicting-same-cc}, all vertices of $U$ belong to the same component $K_i$ of $G+G^*$, for some $i \in [c^*]$.
        By contradiction, suppose that $i \notin [t^*]$. 
        Since $\mathcal{K}^*$ is a typical set of components of $G+G^*$, there exist distinct $i_1, i_2, i_3 \in [t^*]$ such that $K_{i_1}, K_{i_2}, K_{i_3}$ all have the same signature as $K_i$. 
        Let $u \in U$ and let $j \in [p^*]$ such that $u \in P_j$.
        Then, $v_{i_1, j}, v_{i_2, j}, v_{i_3, j} \in f^*(u)$ by \cref{cl:f*-faithful}, so $u$ is flagged.
        Furthermore, $|f^*(u)| \geq 3$ so $\widehat{f^*}(u) = f^*(u)$.
        Thus, $i_1, i_2, i_3 \in \tau^*(U)$, so $|\tau^*(U)| \geq 3$, a contradiction.
    
        Thus, we have $i \in [t^*]$.
        Let $u \in U$ and let $j \in [p^*]$ such that $u \in P_j$.
        Then, $v_{i, j} \in f^*(u)$ by \cref{cl:f*-faithful}. 
        Using that $\phi^*$ is a crossing solver, it then follows from unpacking the definitions of crossing solver and of updated compatibility function that $v_{i, j} \in \widehat{f^*}(u)$.
        Thus, $i \in \tau^*(U)$.
    \end{subproof}

    \paragraph*{Suitable partition.}
    The next step will be to solve an auxiliary problem to try to compute the partition into components of $G+G^*$.
    We will prove in \cref{cl:C*-cc-G+G'} that if the guesses are $(t^*, p^*, \mathcal{R}^*, R^*, G^*/\mathcal{P}^*)$ then the partition we compute is (roughly) the partition into components of some rank-$k$ perturbation $G'$ of $G$ that minimizes the largest component size.
    
    The \emph{width} of a partition $\mathcal{C}$ is the maximum size of a part of $\mathcal{C}$.
    Say that a partition $\mathcal{C}$ of $V(G)$ is \emph{suitable} if it satisfies the following constraints: \begin{itemize}
        \item $\mathcal{U}$ refines $\mathcal{C}$, 
        \item For every $i \in [t]$ there exists $C_i \in \mathcal{C}$ such that $v_{i, j} \in C_i$ for every $j \in [p]$ such that $(i, j) \in \mathcal{R}$.
        \item If $U \in \mathcal{U}$ is such that $|\tau(U)| \leq 2$ or some $u \in U$ is not flagged then there exists $i \in \tau(U)$ such that $U \subseteq C_i$.
    \end{itemize}
    The algorithm computes a suitable partition $\mathcal{C}$ of $V(G)$ of minimum width.
    If there is no such partition, the algorithm aborts. 
    Note that the partition with a single part equal to $V(G)$ is always suitable, unless there exists $U \in \mathcal{U}$ with $\tau(U) = \emptyset$, in which case there is no suitable partition. 
    Thus, the algorithm only aborts if there exists $U \in \mathcal{U}$ with $\tau(U) = \emptyset$.
    Note that in the case where we guessed $(t^*, p^*, \mathcal{R}^*, R^*, G^*/\mathcal{P}^*)$ and $\phi^*$, the partition into components of $G+G^*$ is a suitable partition by \cref{cl:conflicting-same-cc}, by definition of $R^*$ and by \cref{cl:updated-compatibility-info}.

    We now argue that such a partition $\mathcal{C}$ can be computed efficiently.
    Note that there is always a partition of minimum width where $C_i \neq C_{i'}$ for all $i \neq i'$, so we can restrict our attention to such partitions.
    Similarly, there is always a partition of minimum width where every $U \in \mathcal{U}$ such that $|\tau(U)| \geq 3$ and every $u \in U$ is flagged forms its own part of $\mathcal{C}$, so we can restrict our attention to such partitions. 
    Then, denote the other parts of $\mathcal{U}$ by $U_1, \ldots, U_r$.
    For each $i \in [r]$, let $w_i \coloneqq |U_i|$ and $L_i \coloneqq \tau(U_i) \subseteq [t]$.
    Note that $\sum_{i = 1}^r|U_i| \leq n$.
    Observe now that our problem is simply an instance of \textsc{weighted list balancing}, as we want to find an assignment $\alpha : [r] \to [t]$ of the $U_i$ to $C_1, \ldots, C_t$ such that $\alpha(i) \in L_i$ for every $i \in [r]$.
    Note that the fact that $v_{i, j} \in C_i$ for all $v_{i, j} \in R$ will follow from the fact that $f(v_{i, j}) = \{v_{i, j}\}$ so the part $U \in \mathcal{U}$ that contains $v_{i, j}$ satisfies $\tau(U) \subseteq \{i\}$.
    By \cref{lem:list-balancing-easy}, we can then find the desired suitable partition $\mathcal{C}$ in time $\mathcal{O}(t\cdot |\mathcal{U}| \cdot |V(G)|^{t-1}) = n^{f(k)}$ for some computable function $f$.

    \paragraph*{Perturbation.}

    We would now like to say that $\mathcal{C}$ can be realized as the partition into components of some rank-$k$ perturbation $G+\widehat{G}$ of $G$. 
    Note that since we are given $G$ and the partition $\mathcal{C}$, this uniquely determines the edges of $\widehat{G}$ between vertices of $G$ that belong to different parts of $\mathcal{C}$.
    Thus, it suffices to figure out whether there is a way to choose the edges inside each part of $\mathcal{C}$ so that the resulting looped graph has rank at most $k$.
    It follows from the work of Peeters \cite{P96} that the general problem of deciding whether a partial binary matrix can be completed to a binary matrix of rank at most $k$ is already \NP-hard for $k=3$, and we do not see how to solve this problem more efficiently in our context.

    However, with all the information that we have gathered up to this point, there is a natural way to define a candidate looped graph $\widehat{G}$ of rank at most $k$ such that the partition into components of $G+\widehat{G}$ might refine $\mathcal{C}$.

    We define a function $\psi : V(G) \to [p]$.
    We will then define $\widehat{G}$ as the graph on vertex set $V(G)$ where $u, u' \in V(G)$ are adjacent if and only if $\psi(u)$ and $\psi(u')$ are adjacent in $H$.
    Since $H$ has rank at most $k$, it follows from \cref{obs:rank-quotient} that $\widehat{G}$ has rank at most $k$.

    Let $u \in V(G)$ and let $U \in \mathcal{U}$ such that $u \in U$.
    Suppose first that $|\tau(U)| \geq 3$ and every $u \in U$ is flagged.
    Then, by \cref{cl:flagged-unique}, there exists a unique $j \in [p]$ such that $|\{i \in [t] : v_{i, j} \in f(u)\}| \geq 2$.
    We set $\psi(u) = j$.
    Otherwise, by assumption there exists $i \in \tau(U)$ such that $U \subseteq C_i$.
    Since $i \in \tau(U)$, there exists $j \in [p]$ such that $v_{i, j} \in \widehat{f}(u) \subseteq f(u)$. 
    Such a $j$ is unique since $f$ is functional.
    We set $\psi(u) = j$.

    The algorithm then verifies that the partition into components of $G+\widehat{G}$ refines $\mathcal{C}$.
    If so, it returns the width of $\mathcal{C}$, and if not it aborts.
    Then, the algorithm outputs the minimum width that was returned among all guesses for which the corresponding run did not abort.

    \paragraph*{Correctness.}

    Let $OPT$ be the solution of the instance $(G, k)$, so $OPT$ is the maximum size of a component of $G+G^*$ since $G^*$ is $k$-optimal for $G$.
    Let $w$ be the value output by the algorithm. 
    We now argue that $w = OPT$.

    First, observe that for every guess, if the algorithm returns some value $v$ then there is a rank-$k$ perturbation of $G$ whose maximum component has size at most $v$.

    Consider now the run of the algorithm where the guesses are $(t^*, p^*, \mathcal{R}^*, R^*, G^*/\mathcal{P}^*)$ and $\phi^*$.
    As before, let $f^*, \mathcal{U}^*, X^*, \widehat{f^*}, \tau^*, \mathcal{C}^*, \psi^*$ be the values computed by the algorithm during this run.
    Note that $R^*$ is consistent with $G^*/\mathcal{P}^*$ by \cref{cl:consistent}, $f^*$ is faithful and functional by \cref{cl:f*-faithful,f*-functional} and the partition into components of $G+G^*$ is a candidate for the partition $\mathcal{C}^*$ as we already argued. 
    Thus, this run does not abort before computing $\widehat{G}$, and the width of the partition $\mathcal{C}^*$ is at most $OPT$. We now show that this run does not abort at all.

    \begin{claim}\label{cl:C*-cc-G+G'}
        The partition into components of $G+\widehat{G}$ refines $\mathcal{C}^*$.
    \end{claim}

    \begin{subproof}{\cref{cl:C*-cc-G+G'}}
        It suffices to prove that any two vertices that are adjacent in $G+\widehat{G}$ are in the same part of $\mathcal{C}^*$.
        Let $uu' \in E(G+\widehat{G})$, let $U \in \mathcal{U}^*$ (resp. $U' \in \mathcal{U}^*$) be the part of $\mathcal{U}^*$ that contains $u$ (resp. $u'$), and let $j = \psi^*(u)$ and $j' = \psi^*(u')$.
        We show that $u$ and $u'$ belong to the same part of $\mathcal{C}^*$.
        Note that $\mathcal{U}^*$ refines $\mathcal{C}^*$ by definition of $\mathcal{C}^*$, so if $U = U'$ then we are done. 
        From now on, we assume that $U \neq U'$, so $u$ and $u'$ are not strongly conflicting.
        
        Suppose first that $|\tau^*(U)| \geq 3$ and every $v \in U$ is flagged.
        Since $|\tau^*(U)| \geq 3$, we have $|f^*(u)| \geq 3$.
        Then, $u$ and $u'$ are not weakly conflicting by \cref{cl:setup-weakly-conflicting}.
        By definition of $\psi^*$, we have $|\{i \in [t^*] : v_{i, j} \in f^*(u)\}| \geq 2$ and there exists $i' \in [t^*]$ such that $v_{i', j'} \in f^*(u')$.
        Thus, there exist $i \neq i' \in [t^*]$ such that $v_{i, j} \in f^*(u)$ and $v_{i', j'} \in f^*(u')$.
        Using that $u$ and $u'$ are not conflicting, that $R^*$ is consistent with $G^*/\mathcal{P}^*$ and the definition of $\widehat{G}$, we have: \begin{align*}
            uu' \in E(G) &\iff v_{i, j}v_{i', j'} \in E(G) \\
                &\iff jj' \in E(G^*/\mathcal{P}^*) \\
                &\iff uu' \in E(\widehat{G}).
        \end{align*}
        This contradicts that $uu' \in E(G+\widehat{G})$.
    
        By symmetry, it is also not the case that $|\tau^*(U')| \geq 3$ and every $v' \in U'$ is flagged.
        Thus, by definition of $\mathcal{C}^*$, there exist $i \in \tau^*(U)$ and $i' \in \tau^*(U')$ such that $U \subseteq C_i$ and $U' \subseteq C_{i'}$.
        By definition of $\psi^*$, we have $v_{i, j} \in f^*(u)$ and $v_{i', j'} \in f^*(u')$.
        By contradiction, suppose $i \neq i'$.
        Consider first the case where $u$ and $u'$ are not conflicting.
        Then, using also that $R$ is consistent with $H$ and the definition of $\widehat{G}$, we have: \begin{align*}
            uu' \in E(G) &\iff v_{i, j}v_{i', j'} \in E(G) \\
                &\iff jj' \in E(G^*/\mathcal{P}^*) \\
                &\iff uu' \in E(\widehat{G}).
        \end{align*}
        Again, this contradicts that $uu' \in E(G+\widehat{G})$.
        Consider now the case where $u$ and $u'$ are conflicting, in which case they must be weakly conflicting.
        By \cref{cl:setup-weakly-conflicting}, there exists $x = (i_1, i_2, j_1, j_2, j'_1, j'_2) \in X^*$ such that $u \in W_x$ and $u' \in W'_x$, say.
        Recall that $D_x$ denotes the crossing digraph for $x$.
        Suppose first that $\phi^*(x) = \bot$. 
        By definition of $\widehat{f^*}$, we have $\widehat{f^*}(u) \subseteq \{v_{i_2, j_2}\}$ and $\widehat{f^*}(u') \subseteq \{v_{i_2, j'_2}\}$. 
        But then $\tau^*(U), \tau^*(U') \subseteq \{i_2\}$ so $i = i_2 = i'$, a contradiction.
        
        Otherwise, we have $\phi^*(x) \in W_x \cup W'_x$.
        By symmetry between $u$ and $u'$, we may assume that $\phi^*(x) \in W_x$. 
        Note that any vertex of $W_x$ and any vertex of $W'_x$ are connected by an arc in $D_x$.
        Then, either $\phi^*(x)$ can reach $u'$ in $D_x$ or $u'$ can reach $\phi^*(x)$ in $D_x$.
        By definition of $\widehat{f}$, in the first case we have $\tau^*(U') \subseteq \{i_1\}$, and if the first case does not hold we have $\tau^*(U') \subseteq \{i_2\}$.
        
        Suppose that $\phi^*(x)$ can reach $u'$ in $D_x$, so $\tau^*(U') \subseteq \{i_1\}$ so $i' = i_1$ and $j' = j'_1$. 
        Since $i \neq i'$, we have $i = i_2$ and $j = j_2$.
        Since $x$ is crossing, we have $v_{i_1, j'_1}v_{i_2, j_2} \notin E(G)$ so $j'_1j_2 \notin E(G^*/\mathcal{P}^*)$ since $R^*$ is consistent with $G^*/\mathcal{P}^*$ by \cref{cl:consistent}. 
        Thus, $uu' \notin E(\widehat{G})$ by definition of $\widehat{G}$, so $uu' \in E(G)$. 
        Thus, there is an arc from $u'$ to $u$ in $D_x$, so $\phi^*(x)$ can reach $u$ in $D_x$ so $\tau^*(U) \subseteq \{i_1\}$ and $i = i_1$, a contradiction.
        
        Suppose finally that $\phi^*(x)$ cannot reach $u'$ in $D_x$, so $u'$ can reach $\phi^*(x)$ in $D_x$, so $\tau^*(U') \subseteq \{i_2\}$ so $i' = i_2$ and $j' = j'_2$. 
        Since $i \neq i'$, we have $i = i_1$ and $j = j_1$.
        Since $x$ is crossing, we have $v_{i_1, j_1}v_{i_2, j'_2} \in E(G)$ so $j_1j'_2 \in E(G^*/\mathcal{P}^*)$ since $R^*$ is consistent with $G^*/\mathcal{P}^*$ by \cref{cl:consistent}. 
        Thus, $uu' \in E(\widehat{G})$ by definition of $\widehat{G}$, so $uu' \notin E(G)$. 
        Thus, there is an arc from $u$ to $u'$ in $D_x$, so $u$ can reach $u'$ in $D_x$.
        Furthermore, $\phi^*(x)$ cannot reach $u$ in $D_x$, otherwise $\phi^*(x)$ could reach $u'$ in $D_x$.
        Thus, $\tau^*(U) \subseteq \{i_2\}$ and $i=i_2$, a contradiction.
        We obtained a contradiction in all cases, so we must have $i = i'$, so $u \in C_i = C_{i'} \ni u'$, as desired.
    \end{subproof}

    Therefore, when the guesses are $(t^*, p^*, \mathcal{R}^*, R^*, G^*/\mathcal{P}^*)$ and $\phi^*$, the run returns a value $v \leq OPT$, so $w \leq OPT$.
    This concludes the proof that $w = OPT$, which proves the correctness of the algorithm.

    \paragraph*{Running time.}

    Let $n$ be the number of vertices of $G$.
    It follows from the description and analysis of the algorithm that each run can be performed in time $n^{f(k)}$ for some (computable) function $f$ (actually, all steps can be performed in \FPT{} time except the computation of $\mathcal{C}$).
    The number of possible choices for $t, p, \mathcal{R}$ and $H$ is bounded by a function of $k$.
    Thus, the size of $R$ is bounded by a function of $k$.
    The size of $X$ is also bounded by a function of $k$ so the ``size'' of $\phi$ is also bounded by a function of $k$.
    Thus, the total number of runs of the algorithm is $n^{g(k)}$ for some (computable) function $g$.
    Thus, this algorithm is indeed an \XP{} algorithm solving \textsc{rank integrity}.
\end{proof}

\section{\textsc{Rank integrity} is \texorpdfstring{\W$[1]$}{W[1]}-hard}
\label{sec:hardness}

\subsection{Some tools}

We start with some preliminary results.
We view $\{0, 1\}^n$ as a vector space over $\text{GF}(2)$.
For every $n \geq 1$, let $H_n$ be the graph on vertex set $\{0, 1\}^n$ where $x, y \in \{0, 1\}^n$ are adjacent if and only if $x \cdot y = 0$ (where the scalar product is done in $\text{GF}(2)$).
Note that the adjacency matrix of $H_n$ is closely related to the $n$-th Sylvester-Hadamard matrix.

Given two disjoint subsets $X, Y$ of vertices of a graph $G$, we denote by $G[X, Y]$ the bipartite graph on vertex set $X \cup Y$ with bipartition $(X, Y)$, and where the edges between $X$ and $Y$ are the same as in $G$.

\begin{lemma}\label{lem:small-complete-empty}
    Let $n \geq 1$ and let $X, Y$ be disjoint subsets of $\{0, 1\}^n$ such that the bipartite graph $H_n[X, Y]$ is complete or empty.
    Then, $|X| \cdot |Y| \leq 2^{n}$.
\end{lemma}

\begin{proof}
    Since $H_n[X, Y]$ is complete or empty, for all $x, x' \in X$ and $y, y' \in Y$ we have $x \cdot y = x' \cdot y'$.
    If $X = \emptyset$ then we are done, so suppose $X \neq \emptyset$.
    Fix $x_0 \in X$ arbitrarily. 
    Then, for all $x \in X$ and $y \in Y$ we have $(x-x_0)\cdot y = x_0 \cdot y - x \cdot y = 0$.
    
    Let $X - x_0 \coloneqq \{x - x_0 : x \in X\} \subseteq \{0, 1\}^n$ and let $U \coloneqq \spn(X-x_0) \subseteq \{0, 1\}^n$.
    We showed that $Y \subseteq (X - x_0)^{\perp} = U^{\perp}$.
    The scalar product is a non-degenerate bilinear form over $\{0, 1\}^n$ so $\dim(U) + \dim(\spn(Y)) \leq \dim(U) + \dim(U^{\perp}) = n$.
    Therefore, $|X|\cdot|Y| \leq 2^{\dim(U)} \cdot 2^{n - \dim(U)} = 2^n$.
\end{proof}

A \emph{cut} of a graph $G$ is a partition $(A, B)$ of $V(G)$.
The \emph{rank} of the cut $(A, B)$ is the rank of the adjacency matrix between $A$ and $B$ in $G$.

\begin{lemma}\label{lem:small-rank-small-side}
    Let $n \geq 1$ and let $(A, B)$ be a cut of $H_n$ of rank at most $k$.
    Then, $\min(|A|, |B|) \leq 2^{2k+1}$.
\end{lemma}

\begin{proof}
    Since the cut $(A, B)$ has rank at most $k$, we can partition $A$ into $r \leq 2^k$ parts $A_1, \ldots, A_{r}$ and $B$ into $s \leq 2^k$ parts $B_1, \ldots, B_s$ such that each bipartite graph $H_n[A_i, B_j]$ is complete or empty.
    Thus, \cref{lem:small-complete-empty} implies that $|A_i| \cdot |B_j| \leq 2^n$ for all $i \in [r]$ and $j \in [s]$.
    Altogether, this implies $|A| \cdot |B| \leq 2^{2k} \cdot 2^n$.
    Without loss of generality, we may assume that $|A| \leq |B|$ so $|B| \geq 2^{n-1}$ since $(A, B)$ is a partition of $\{0, 1\}^n$.
    Then, $|A| \cdot |B| \geq |A| \cdot 2^{n-1}$ so $|A| \cdot 2^{n-1} \leq 2^{2k} \cdot 2^n$ so $|A| \leq 2^{2k+1}$.
\end{proof}

We now show that the graph $H_n$ is robustly connected with respect to perturbations of small rank.

\begin{lemma}\label{lem:exists-giant-comp}
    Let $n \geq 2k+3$ and let $H_n+\widehat{H}$ be a rank-$k$ perturbation of $H_n$.
    Then, $H_n + \widehat{H}$ has a component of size at least $2^n-2^{2k+1}$.
\end{lemma}

\begin{proof}
    By contradiction, suppose that all the components of $H_n + \widehat{H}$ have less than $2^n-2^{2k+1}$ vertices.
    If one component has more than $2^{n-2}$ vertices, let $A$ be the vertex set of such a component. Then, $|A|, |V(H_n) \setminus A| > 2^{2k+1}$ and there is no edge between $A$ and $V(H_n) \setminus A$ in $H_n + \widehat{H}$.
    If all components have size at most $2^{n-2}$, by considering a minimal collection of components whose union has size at least $2^{n-1}$, we obtain a set $A$ satisfying $2^{n}/2 \leq |A| <3 \cdot 2^n/4$. 
    Thus, $|A|, |V(H_n) \setminus A| > 2^{2k+1}$ and there is no edge between $A$ and $V(H_n) \setminus A$ in $H_n + \widehat{H}$.
    In both cases, the graph $(H_n + \widehat{H})[A, V(H_n) \setminus A]$ is edgeless, so $H_n[A, V(H_n) \setminus A] = \widehat{H}[A, V(H_n) \setminus A]$, which implies that $H_n[A, V(H_n) \setminus A]$ has rank at most $k$.
    Thus, $(A, V(H_n) \setminus A)$ is a cut of $H_n$ of rank at most $k$ with $\min(|A|, |V(H_n) \setminus A|) > 2^{2k+1}$, which contradicts \cref{lem:small-rank-small-side}.
\end{proof}

For a graph $G$, the \emph{vertex-edge incidence matrix} of $G$ is the matrix $I_G$ whose rows are indexed by $V(G)$, whose columns are indexed by $E(G)$, and such that $I_G[v, e] = 1$ if $v$ and $e$ are incident and $I_G[v, e] = 0$ otherwise.
The \emph{vertex-edge incidence graph} of $G$ is the bipartite graph on vertex set $V(G) \cup E(G)$ with bipartition $(V(G), E(G))$, where $v \in V(G)$ and $e \in E(G)$ are adjacent if and only if they are incident in $G$.
Note that this graph is simply the $1$-subdivision of $G$.
Observe that the bipartite adjacency matrix of this graph is the vertex-edge incidence matrix of $G$.

It is well-known that, for all $k \geq 1$, the vertex-edge incidence matrix of $K_k$ has rank $k-1$.
The next lemma essentially says that there is a looped graph $G_k$ of rank $k-1$ on vertex set $V(K_k) \cup E(K_k)$ such that $G_k[V(K_k), E(K_k)]$ is the vertex-edge incidence graph of $K_k$.
In other words, the vertex-edge incidence graph of $K_k$, which has ``bipartite rank'' $k-1$, can be completed to a looped graph of rank $k-1$.

\begin{lemma}\label{lem:rk-Kk1}
    Let $k \geq 1$ and let $B_k$ be the vertex-edge incidence matrix of $K_k$.
    There exist symmetric matrices $A_k$ and $C_k$ such that the matrix \[ M_k\coloneqq
\begin{bmatrix}
A_k & B_k \\
(B_k)^{\intercal} & C_k
\end{bmatrix}\] has rank $k-1$.
\end{lemma}

\begin{proof}
    Let $e_1, \ldots, e_{k-1}$ denote the canonical basis of $\{0, 1\}^{k-1}$.
    For every $i \in [k-1]$, let $u_i \coloneqq e_i$, and let $u_k \coloneqq e_1 + \ldots + e_k$.
    For all $i < j \in [k-1]$, let $u_{i, j} \coloneqq e_i + e_j$, and for every $i \in [k-1]$, let $u_{i, k} \coloneqq e_i$.
    Let $U_1 \coloneqq \begin{bmatrix}
        u_1 \ldots u_k
    \end{bmatrix} \in (\{0, 1\}^{k-1})^{k}$, $U_2 \coloneqq \begin{bmatrix}
        u_{1, 2} \ldots u_{k-1, k}
    \end{bmatrix} \in (\{0, 1\}^{k-1})^{\binom{k}{2}}$ and $U \coloneqq \begin{bmatrix}
        U_1 U_2
    \end{bmatrix} \in (\{0, 1\}^{k-1})^{k + \binom{k}{2}}$.
    We claim that $(U_1)^{\intercal}U_2 = B_k$.
    Indeed, let $i \in [k]$ and $j < j' \in [k]$. We consider four cases depending on the values of $i$ and $j'$. \begin{itemize}
        \item If $i \in [k-1]$ and $j' \in [k-1]$ then $u_i \cdot u_{j, j'} = e_i \cdot (e_j + e_{j'})$ is $1$ if and only if $i \in \{j, j'\}$.
        \item If $i \in [k-1]$ and $j' = k$ then $u_i \cdot u_{j, j'} = e_i \cdot e_j$ is $1$ if and only if $i \in \{j, j'\}$.
        \item If $i = k$ and $j' \in [k-1]$ then $u_i \cdot u_{j, j'} = (e_1 + \ldots + e_{k-1}) \cdot (e_j + e_{j'}) = 0$.
        \item If $i = k$ and $j' = k$ then $u_i \cdot u_{j, j'} = (e_1 + \ldots + e_{k-1}) \cdot e_j = 1$.
    \end{itemize}
    Let $A_k \coloneqq (U_1)^{\intercal} U_1$ and $C_k \coloneqq (U_2)^{\intercal} U_2$, and note that $A_k$ and $C_k$ are symmetric matrices.
    Note that $U$ consists of vectors of $\{0, 1\}^{k-1}$ so $U$ has rank at most $k-1$.
    Then, setting $M_k \coloneqq U^{\intercal}U$, it follows that $M_k$ has rank at most $k-1$ and \[ M_k=
\begin{bmatrix}
A_k & B_k \\
(B_k)^{\intercal} & C_k
\end{bmatrix}.\]

    Finally, note that the $k-1$ first rows of $B_k$ are linearly independent (because for each row, there is a column in which it has entry $1$, while all the other first $k-1$ rows have entry $0$), so $\rk(M_k) \geq \rk(B_k) \geq k-1$.
\end{proof}

We now prove that the vertex-edge incidence matrix of any graph with at least $\binom{k}{2}$ edges and no clique of size $k$ has rank at least $k$.
For this, we need two more lemmas. The first one is folklore, see e.g. \cite[Exercise 1.4.2]{BM08}.

\begin{lemma}\label{lem:rank-n-c}
    Let $G$ be a graph with $n$ vertices and $c$ components. 
    Let $I_G$ be the vertex-edge incidence matrix of $G$.
    Then, $\rk(I_G) = n-c$.
\end{lemma}

\begin{lemma}\label{lem:n-c}
    Let $k \geq 1$ and let $G$ be a graph with $\geq \binom{k}{2}$ edges, $n$ vertices, $c$ components and no clique of size $k$. Then, $n - c \geq k$.
\end{lemma}

\begin{proof}
    By contradiction, suppose that $n - c \leq k-1$.
    As long as there are (at least) two components of $G$ of size at least $2$, move one vertex from the smallest such component to the largest one, and make it adjacent to all the vertices in its new component. 
    Then, add all the missing edges inside each component. 
    Observe that each modification adds at least one edge.
    Eventually, we obtain the graph $G^*$ with $(c-1)$ components of size $1$ and one component of size $n-(c-1) \leq k$ which forms a clique.
    Note that $G^*$ has exactly $\binom{n-c+1}{2} \leq \binom{k}{2}$ edges, so both $G$ and $G^*$ have exactly $\binom{k}{2}$ edges. 
    Thus, we did not add any edge when going from $G$ to $G^*$, which means that $G = G^*$.
    Since $G^*$ has $\binom{k}{2}$ edges, $G^*$ must contain $K_k$, so $G$ contains a clique of size $k$, a contradiction.
\end{proof}

The result now follows immediately.

\begin{proposition}\label{prop:find-Kk}
    Let $k \geq 1$ and let $G$ be a graph with at least $\binom{k}{2}$ edges and no clique of size $k$. Then, the vertex-edge incidence matrix of $G$ has rank at least $k$.
\end{proposition}

\begin{proof}
    Let $I_G$ denote the vertex-edge incidence matrix of $G$.
    Let $n$ denote the number of vertices of $G$ and $c$ its number of components.
    Combining \cref{lem:rank-n-c} and \cref{lem:n-c}, we get $\rk(I_G) = n-c \geq k$.
\end{proof}

\subsection{The reduction}

We now describe the construction that we will use in our reduction.
Let $G = (V, E)$ be a graph and let $k$ be an integer.
We define a new graph $\Gamma_k(G)$ as follows. \begin{itemize}
    \item $V(\Gamma_k(G)) = (V \times [2^{2k}]) \cup (E \times [2^{2k}]) \cup D$, where $D$ is a set of (dummy) vertices of size at most $\max\{|V| \cdot 2^{2k}, \binom{k}{2} \cdot 2^{2k+2}\}$ such that $|(V \times [2^{2k}]) \cup D|$ is a power of two, say $2^p$, with $2^p \geq \binom{k}{2} \cdot 2^{2k+1}$.
    \item $\Gamma_k(G)[(V \times [2^{2k}]) \cup D]$ is isomorphic to the graph $H_p$.
    \item $\Gamma_k(G)[(E \times [2^{2k}]) \cup D]$ is edgeless.
    \item There is an edge in $\Gamma_k(G)$ between $(v, \cdot)$ and $(e, \cdot)$ for $v \in V$ and $e \in E$ if and only if $v$ and $e$ are incident in $G$. Thus, $\Gamma_k(G)[V \times [2^{2k}], E \times [2^{2k}]]$ can be obtained from the vertex-edge incidence graph of $G$ by replacing each vertex by a stable set of size $2^{2k}$.
\end{itemize}
For every vertex $v \in V(G)$, let $V_v$ be the set of vertices of $\Gamma_k(G)$ of the form $(v, \cdot)$, and for every edge $e \in E(G)$, let $V_e$ be the set of vertices of $\Gamma_k(G)$ of the form $(e, \cdot)$.

\begin{lemma}\label{lem:pf-reduction}
    Let $k \geq 1$ and let $G$ be a graph.
    Then, $G$ contains a clique of size $k$ if and only if there exists a rank-$(k-1)$ perturbation of $\Gamma_k(G)$ whose largest component has size at most $|\Gamma_k(G)| - \binom{k}{2} \cdot 2^{2k}$.
\end{lemma}

\begin{proof}
    Suppose first that $G$ contains a clique $C$ of size $k$.
    By \cref{lem:rk-Kk1}, there exists a looped graph $G_k$ of rank $k-1$ on vertex set $V(C) \cup E(C)$ such that $G_k[V(C), E(C)]$ is isomorphic to the vertex-edge incidence graph between $V(C)$ and $E(C)$ in $G$.
    By replacing each vertex in $V(C) \cup E(C)$ by a clique or a stable set of size $2^{2k}$, according to whether the corresponding vertex is looped or not, it follows that there exists a graph $G'_k$ of rank $k-1$ on vertex set $\bigcup_{v \in V(C)} V_v \cup \bigcup_{e \in E(C)} V_e$ such that $G'_k[\bigcup_{v \in V(C)} V_v, \bigcup_{e \in E(C)} V_e] = \Gamma_k(G)[\bigcup_{v \in V(C)} V_v, \bigcup_{e \in E(C)} V_e]$.
    Let $\widehat{\Gamma}$ be the graph on vertex set $V(\Gamma_k(G))$ obtained from $G'_k$ by adding the vertices of $V(\Gamma_k(G)) \setminus V(G'_k)$ as isolated vertices.
    Then, $\widehat{\Gamma}$ is a graph of rank $k-1$ on the same vertex set as $\Gamma_k(G)$. 
    Furthermore, the vertices in $\bigcup_{e \in E(C)} V_e$ are separated from the other vertices in $\Gamma_k(G) + \widehat{\Gamma}$, so the largest component of $\Gamma_k(G)+\widehat{\Gamma}$ has size at most $|\Gamma_k(G)| - \binom{k}{2} \cdot 2^{2k}$. 
    Note that we are using here that $|\Gamma_k(G)| \geq |(V \times [2^{2k}]) \cup D| \geq \binom{k}{2} \cdot 2^{2k+1}$, so $\left|\bigcup_{e \in E(C)} V_e\right| \leq |\Gamma_k(G)| - \binom{k}{2} \cdot 2^{2k}$.

    Conversely, suppose that there exists a rank-$(k-1)$ perturbation $\Gamma_k(G) + \widehat{\Gamma}$ of $\Gamma_k(G)$ whose largest component has size at most $|\Gamma_k(G)| - \binom{k}{2} \cdot 2^{2k}$.
    Note that $2^p \geq 2^{2k+1}$ by construction so $p \geq 2k+1 = 2(k-1) + 3$.
    By \cref{lem:exists-giant-comp}, some component $K$ of $(\Gamma_k(G) + \widehat{\Gamma})[(V \times [2^{2k}]) \cup D]$ has size at least $2^p - 2^{2k-1}$.
    Therefore, $K$ intersects all sets $V_v$ for $v \in V(G)$.
    Let $K'$ be the component of $\Gamma_k(G) + \widehat{\Gamma}$ that contains $K$.
    Then, $K'$ misses at least $\binom{k}{2} \cdot 2^{2k} - 2^{2k-1} > \left(\binom{k}{2} - 1\right) \cdot 2^{2k}$ vertices in $\bigcup_{e \in E(G)}V_e$.
    Thus, $K'$ must miss at least one vertex from at least $\binom{k}{2}$ sets $V_e$.
    Let $F \subseteq E(G)$ be the set of all edges $e \in E(G)$ such that $K'$ misses at least one vertex from $V_e$. 
    Note that $|F| \geq \binom{k}{2}$.

    By contradiction, suppose that $G[F]$ does not contain a clique of size $k$.
    Let $I_{G[F]}$ denote the vertex-edge incidence matrix of $G[F]$. 
    By \cref{prop:find-Kk}, we have $\rk(I_{G[F]}) \geq k$.
    Since $K'$ intersects all sets $V_v$ for $v \in V(G)$ and misses at least one vertex from each set $V_e$ for $e \in F$, it follows that $\widehat{\Gamma}$ contains the vertex-edge incidence graph of $G[F]$ as a (bipartite) semi-induced subgraph. 
    Thus, the adjacency matrix of $\widehat{\Gamma}$ contains $I_{G[F]}$ as a submatrix, so $rk(\widehat{\Gamma}) \geq rk(I_{G[F]}) \geq k$, a contradiction.
    Therefore, $G[F]$ contains a clique of size $k$, so $G$ contains a clique of size $k$.
\end{proof}

It is now easy to deduce \cref{thm:W1-hard} from \cref{lem:pf-reduction}. We restate it for convenience.

\thmHardness*

\begin{proof}
    Given an instance $(G, k)$ of \textsc{clique}, we construct the graph $\Gamma_k(G)$, which has size $|\Gamma_k(G)| \leq \max\{2n \cdot 2^{2k}, n \cdot 2^{2k} + \binom{k}{2} \cdot 2^{2k+2}\} + \binom{n}{2} \cdot 2^{2k}$, and consider the instance $(\Gamma_k(G), k-1)$ of \textsc{rank integrity}.
    By \cref{lem:pf-reduction}, the solution to this instance is at most $|\Gamma_k(G)| - \binom{k}{2} \cdot 2^{2k}$ if and only if $G$ contains a clique of size $k$.
\end{proof}

\section{Computing 1-ancilla-integrity}
\label{sec:ancillaAlgorithm}

In this section we prove \Cref{thm:ancillaVuln} about computing the $1$-ancilla-integrity of a graph. 

First we show how to reduce this problem to computing flip-integrity. Recall that the \emph{flip-integrity} of a graph $G$ is the minimum, over all sets $S \subseteq V(G)$, of the maximum component size of the graph $G \circ S$ which is obtained from $G$ by flipping on~$S$. For convenience, given a graph $G$, we say that an \emph{optimal flip set} is a set $S^* \subseteq V(G)$ so that the largest size of a component of $G \circ S^*$ equals the flip-integrity of~$G$. Recall that, given a graph $G$ and a vertex $v$ of $G$, we write $G*v$ for the graph obtained from $G$ by locally complementing at $v$.

\begin{lemma}
\label{lem:reductionAncilla}
    The $1$-ancilla-integrity of a graph $G$ equals the minimum, over all graphs $G'$ obtained from $G$ by performing at most one local complementation, of the flip-integrity of~$G'$.
\end{lemma}
\begin{proof}
    First we prove that the $1$-ancilla-integrity of $G$ is at most the claimed value. There are two cases. For the first case, suppose that the minimum is achieved by the graph $G$ itself, and let $S^* \subseteq V(G)$ be an optimal flip set of $G$. Then $G$ and $G \circ S^*$ are $1$-perturbations of each other as we can add a vertex whose neighborhood is $S^*$ to $G$ and then locally complement and delete it to form $G \circ S^*$. This finishes the first case.
    
    Next suppose that the minimum is achieved by the graph $G*u$ for some vertex $u$ of $G$, and let $S^* \subseteq V(G)$ be an optimal flip set of $G*u$. Let $\widehat{G}$ be the graph which is obtained from $G*u$ by adding a new vertex $a$ whose neighborhood is $S^*$. Then $G$ is obtained from $\widehat{G}$ by deleting $a$ and then locally complementing on $u$. Moreover, $(G*u) \circ S^*$ is obtained from $\widehat{G}$ by locally complementing on $a$ and then deleting $a$. So $G$ and $(G*u) \circ S^*$ are $1$-perturbations of each other, which finishes the second case.
    
    It just remains to show that the claimed value is at most the $1$-ancilla-integrity of $G$. So, consider a graph $G^*$ which is a $1$-perturbation of $G$ which minimizes the maximum component size. Thus there exists a graph $\widehat{G}$ with one additional vertex so that both $G$ and $G^*$ are vertex-minors of $\widehat{G}$. By locally complementing in $\widehat{G}$, we may assume that $G$ is an induced subgraph of $\widehat{G}$, that is, there exists a vertex $a$ so that $G=\widehat{G}-a$. By \Cref{lem:threeWaysToRemove}, the graph $G^*$ is a vertex-minor of at least one of the following three graphs:\begin{enumerate}
    \item the graph obtained from $\widehat{G}$ by deleting $a$,
    \item the graph obtained from $\widehat{G}$ by locally complementing at $a$ and then deleting $a$, or
    \item the graph obtained from $\widehat{G}$ by selecting an arbitrary neighbor $u$ of $a$, locally complementing on $u$ then $a$ then $u$ again, and finally deleting $a$.
    \end{enumerate}
    In the first case, $G^*=G$, which is trivially a flip of $G$. In the second case, $G^*$ is also a flip of $G$, as desired. So we may assume that $G^* = \widehat{G}*u*a*u-a$ for some neighbor $u$ of $a$.

    Since locally complementing at $u$ does not change the maximum component size of $G^*$, it suffices to consider the graph $\widehat{G}*u*a-a$. This graph is a flip of the graph $G*u$, as desired. This completes the proof.
\end{proof}

\subsection{Reductions for flip-integrity}
\label{sec:triangleSubroutine}

We now turn our attention to computing the flip-integrity of a graph. This subsection is dedicated to giving several reductions. We write $|G|$ for the number of vertices of a graph~$G$.

First we show how to reduce to the case that $G$ is connected. 

\begin{lemma}
\label{lem:makeGConn}
Let $G$ be a disconnected graph with components $C_1, C_2, \ldots, C_r$ ordered so that $|C_1| \geq |C_2| \geq \ldots \geq |C_r|$. Then the flip-integrity of $G$ is the maximum of the flip-integrity of $C_1$ and the size of $C_2$.
\end{lemma}
\begin{proof}
    Certainly the flip-integrity of $G$ is at most the desired quantity since we can select $S$ to be an optimal flip-set of $C_1$. Also notice that the flip-integrity of $G$ is at least the flip-integrity of $C_1$ (or indeed, of any of its induced subgraphs). Thus it suffices to show that for any $S \subseteq V(G)$, there exists a component of $G \circ S$ of size at least $|C_2|$.

    Let $S \subseteq V(G)$. If $S$ does not contain any vertex of $C_2$, then $C_2$ is a component of $G\circ S$, and we are done. So we may assume that there exists a vertex $v$ which is in both $S$ and $C_2$. Now, let $H_1, \ldots, H_t$ be the components of $C_1 \circ (S \cap V(C_1))$. We may assume that $t \geq 2$ since otherwise $|H_1|= |C_1|\geq |C_2|$, and we are done. Since $t \geq 2$, every component $H_i$ contains a vertex $v_i \in S$. Then $v$ and $v_i$ are adjacent in $G \circ S$. It follows that all of the vertices of $C_1$ are in the same component of $G \circ S$. This completes the proof as $|C_1| \geq |C_2|$.
\end{proof}

Next we note that if $|G| \geq 2$, then for every optimal flip set $S^*$, the graph $G \circ S^*$ has at least two components. This is a corollary of the following lemma.

\begin{lemma}
\label{lem:lessThanN}
The flip-integrity of a graph $G$ with at least two vertices is at most~$|G|-1$.
\end{lemma}
\begin{proof}
    Let $v$ be any vertex of $G$, and let $S$ be the set which contains $v$ and all of its neighbors. Then $v$ is isolated in $G \circ S$, and so the largest component of $G\circ S$ has size at most~$|G|-1$.
\end{proof}

The final lemma of this subsection shows how to ``recognize'' the best flip set $S \subseteq V(G)$ so that $G \circ S$ has at least three components (assuming that $G$ is connected). Given a graph $G$, we write $\mathcal{S}_{\geq 3}(G)$ for the collection of all sets $S \subseteq V(G)$ so that either $S = \emptyset$, or $G \circ S$ has at least three components. We put the empty set in $\mathcal{S}_{\geq 3}(G)$ just to ensure that it is always non-empty. Now we are ready to show the following key lemma. 

\begin{lemma}
\label{lem:constrainedFlipVul}
    There is an algorithm which takes in a connected graph $G$ with $n$ vertices and $m$ edges and returns in time $\mathcal{O}(m^{1.5}n^2)$ a set $S\in \mathcal{S}_{\geq 3}(G)$ which minimizes, among all sets in $\mathcal{S}_{\geq 3}(G)$, the maximum size of a component of $G \circ S$.
\end{lemma}
\begin{proof}
    Let $G$ be a connected graph, and consider a set $S\in \mathcal{S}_{\geq 3}(G)$ which minimizes, among all sets in $\mathcal{S}_{\geq 3}(G)$, the maximum size of a component of $G \circ S$. Let $C_1, C_2, \cdots, C_r$ be the components of $G \circ S$. Then in the original graph $G$, for any vertices $u \in C_i \cap S$ and $v \in C_j \cap S$ where $i$ and $j$ are distinct, $u$ and $v$ must have been adjacent in $G$. This is because there are no edges between $u$ and $v$ in $G \circ S$ and these relationships have been flipped. A similar argument shows that for any vertices $u \in C_i \setminus S$ and $v \in C_j \setminus S$ where $i$ and $j$ are distinct, $u$ and $v$ are nonadjacent in $G$. 

    Thus, since $G \circ S$ has at least three components, $G$ contains a triangle whose vertices $x$, $y$, and $z$ are in three different components of $G \circ S$. Notice that every vertex in $S$ forms a triangle (in $G$) with at least two of the vertices from the triangle $x, y, z$. In fact this characterizes $S$; no vertex outside of $S$ forms such a triangle in $G$. So, by iterating through all triangles $x,y,z$ in $G$ and deriving the set \begin{align*}
    \{u \in V(G) : u \text{ is in a triangle in $G$ with at least two of the vertices }x,y,z\},
    \end{align*} we will consider all sets in $\mathcal{S}_{\geq 3}(G)$ and thus find $S$ (when $S$ is non-empty). Note that we can also efficiently test whether a given set is in $\mathcal{S}_{\geq 3}(G)$.

    \begin{algorithm}
    \caption{The algorithm for \Cref{lem:constrainedFlipVul}.}
    \label{alg:ManyComponents}
    \begin{algorithmic}[1]
    \Function{BestFlipWithManyComponents}{G}
        \State $k \gets |V(G)|$ and $S \gets \emptyset$
        \ForAll{$\{x,y,z\} \subseteq V(G)$ with $xy, yz, xz \in E(G)$}
            \State $S' \gets \{u \in V(G) : u \text{ is in a triangle in $G$ with at least two of the vertices }x,y,z\}$
            \If{$G \circ S'$ has at least three components \textbf{and} the largest one has size $k'<k$}
                \State $k \gets k'$ and $S \gets S'$
            \EndIf
        \EndFor
        \State \Return $S$
    \EndFunction
    \end{algorithmic}
    \end{algorithm}

    This procedure is summarized in Algorithm~\ref{alg:ManyComponents}; we are done with the proof of correctness.
    In terms of time complexity, every graph with $m$ edges has $\mathcal{O}(m^{1.5})$ triangles \cite[Theorem~Section 10.7.2]{MiningDatasets}. For each triangle, we can compute the flip set $S'$ in time $\mathcal{O}(n^2)$. We can then compute $G \circ S'$ and the size of its largest component in time $\mathcal{O}(n^2)$ using Breadth-First Search (BFS). Therefore, the total time complexity is $\mathcal{O}(m^{1.5}n^2)$.
\end{proof}

Together, all of the lemmas in this section will be used to reduce the problem of computing the flip-integrity of a graph to the problem of computing a most balanced split.

\subsection{Finding a most balanced split}
\label{sec:mostBalSplit}

Recall that in a graph $G$, a \emph{cut} is a partition $(A, B)$ of the vertex set into two parts.
We consider non-oriented cuts, in the sense that we make no distinction between the cut $(A, B)$ and the cut $(B, A)$.
A cut $(A, B)$ is a \emph{split} if the edges between $A$ and $B$ form the edge-set of a complete bipartite graph. Said differently, every vertex in $A$ that has a neighbor in $B$ is adjacent to every vertex in $B$ that has a neighbor in $A$. (If there are no edges between $A$ and $B$, then we also consider this to be a split.)
The \emph{width} of a cut $(A, B)$ is $\max\{|A|, |B|\}$.
Observe that a most balanced split in $G$ is a split of minimum width.

This subsection is dedicated to proving the following result.

\begin{restatable}{theorem}{computesplit}\label{thm:compute-most-balanced-split}
    There is an algorithm which takes as input a graph $G$ with $n$ vertices and returns a split of $G$ of minimum width in time $\mathcal{O}(n^2)$.
\end{restatable}

To find a most balanced split of $G$, we rely on split decompositions, which were first introduced by Cunningham~\cite{cunningham82}. Note that a graph may contain an exponential number of splits; for instance if the graph is a clique or star, then every cut is a split. However, the split decomposition is an object of polynomial size which ``displays'' all of the splits of $G$ simultaneously. To describe the split decomposition, we need some more definitions.

A split $(A, B)$ of an $n$-vertex graph is \emph{trivial} if it has width at least $n-1$, that is, if one of $A,B$ has size at most one. Otherwise the split is \emph{nontrivial}. A graph is \emph{prime} if it has no nontrivial split.
Two splits $(A, B)$ and $(C, D)$ \emph{cross} if all four sets $A \cap C, A \cap D, B \cap C, B \cap D$ are nonempty.
A split is \emph{strong} if it does not cross any other split.
Observe that every trivial split is strong. In general, any collection of non-crossing cuts is arranged in a ``tree-like'' fashion in a manner which we now explain for the special case of strong splits.

Let $G$ be a graph and let $T$ be a tree whose leaves are the vertices of $G$.
Let $e$ be an edge of $T$, and let $T_A, T_B$ be the two connected components of $T-e$.
Let $A$ be the set of vertices of $G$ corresponding to the leaves of $T$ that are in $T_A$, and define $B$ analogously.
We say that the cut $(A, B)$ is \emph{induced} by the edge $e$ of $T$.

The \emph{split decomposition} of a graph $G$ is the tree $T$ whose leaves are the vertices of $G$ and such that each strong split of $G$ is induced by a unique edge of $T$, and conversely each edge of $T$ induces a strong split.
Its existence follows from the definition of strong splits. (It is a fairly standard fact in graph theory that such a tree can be formed from any collection of pairwise non-crossing cuts; this fact is used for constructing many different types of decompositions. However, one can also just take the existence of $T$ as a fact from~\cite{linearSplitRevisited12, cunningham82}.)

Let $x$ be an internal node of $T$ and let $T_1, \ldots, T_k$ be the connected components of $T-x$. 
For each $i \in [k]$, let $V_i$ be the set of vertices of $G$ corresponding to the leaves of $T$ that are in $T_i$.
The \emph{quotient graph at $x$} is the graph $G_x$ on vertex set $\{v_1, \ldots, v_k\}$ with an edge between $v_i$ and $v_j$ in $G_x$ if and only if there is an edge between $V_i$ and $V_j$ in $G$, for all $i \neq j \in [k]$.
The \emph{lift} of a split $(A_x, B_x)$ of $G_x$ is the cut $(A, B)$ of $V(G)$ with $A = \bigcup_{i : v_i \in A_x}V_i$ and $B = \bigcup_{i : v_i \in B_x}V_i$.
It is straightforward to see that the lift of a split of $G_x$ is a split of $G$.
Cunningham proved that each quotient graph has one of three forms: it is either prime, a clique, or a star. We call an internal node $x$ \emph{prime} if the quotient $G_x$ is prime.
Note that the clique and the star are the only graphs with the property that every cut is also a split.

The split decomposition of a graph displays all its splits in the following sense.

\begin{lemma}[\cite{linearSplitRevisited12, cunningham82}]
\label{lem:splits-displayed}
    Let $G$ be a graph, let $T$ be the split decomposition of $G$, and let $(A, B)$ be a split of $G$.
    Then, either $(A, B)$ is induced by some edge of $T$, or there exists an internal node $x$ of $T$ such that $G_x$ is not prime, and a split $(A_x, B_x)$ of $G_x$ such that $(A, B)$ is the lift of $(A_x, B_x)$.
\end{lemma}

Dalhaus \cite{linearSplitDahlhaus} and Charbit, de Montgolfier, and Raffinot \cite{linearSplitRevisited12} proved that the split decomposition can be computed in linear time.

\begin{theorem}[\cite{linearSplitRevisited12, linearSplitDahlhaus}]
\label{thm:compute-decomp}
    There is an algorithm which takes as input a graph $G$ with $n$ vertices and $m$ edges and returns the split decomposition of $G$ as well as all quotient graphs in time $\mathcal{O}(n+m)$. Moreover, each internal node is labeled by whether or not it is prime.
\end{theorem}

Note that the above result implies that the sum of the sizes of all the quotient graphs is linear in the size of the input graph. We state this fact formally for later use.
Given a tree $T$, we denote by $I(T)$ the set of internal nodes of $T$. 

\begin{lemma}[\cite{linearSplitRevisited12, linearSplitDahlhaus}]\label{lem:sum-small}
    Let $G$ be a graph and let $T$ be the split decomposition of $G$. Then, \[\sum_{x \in I(T)}|V(G_x)| = \mathcal{O}(n).\]
\end{lemma}

In light of \cref{lem:splits-displayed}, given a graph $G$ with split decomposition $T$, to compute a most balanced split of $G$, it suffices to compute a most balanced split induced by an edge of $T$, and a most balanced split that is the lift of a split of some quotient graph $G_x$.

This is how we proceed to prove \cref{thm:compute-most-balanced-split}, which we restate for convenience.

\computesplit*

\begin{proof}
    Given $G$, we start by computing the split decomposition $T$ of $G$.
    This can be done in time $\mathcal{O}(n^2)$ by \cref{thm:compute-decomp}. Note that $T$ has $\mathcal{O}(n)$ nodes by \Cref{lem:sum-small}.

    We view $T$ as being rooted at an arbitrary node. Given a node $x$ of $T$, let $d(x)$ denote the number of vertices of $G$ which are a descendant of $x$ in $T$. First, by performing bottom-up dynamic programming on $T$, we compute all values of $d(x)$ in time $\mathcal{O}(n^2)$. From there, we determine in time $\mathcal{O}(n^2)$ the minimum width of a split of $G$ induced by an edge of $T$.
    
    Now we consider all internal nodes $x$ of $T$ which are not prime in arbitrary order. Thus $G_x$ is either a clique or a star, and all cuts of $G_x$ are splits. Note that we can use the values of $d(\cdot)$ to compute, for each vertex $u$ of $G_x$, how many vertices $w(u)$ of $G$ are in the subtree of $T-x$ corresponding to $u$. We would like to compute a partition of $V(G_x)$ into two parts which minimizes the maximum weight of a part. Write $V(G_x) = \{u_1, \ldots, u_{t}\}$. The problem we are trying to solve is an instance of \textsc{weighted list balancing} with $m = 2$ and items $u_1, \ldots, u_{t}$, each with weight $w(u_i)$ and list $L_i = \{1, 2\}$. By \cref{lem:list-balancing-easy}, this can be done in time $\mathcal{O}(t \cdot n)$. In this manner, we compute the minimum width of a lift of a split of $G_x$. Running this procedure on all internal nodes $x$ of $T$ takes time $\mathcal{O}(n^2)$ by \cref{lem:sum-small}.

    At this point, we have computed at most $|I(T)| + 1 = \mathcal{O}(n)$ widths of splits of $G$ (including minimum width of a split of $G$ induced by an edge of $T$). We simply return the minimum among all of them. The correctness follows immediately from \cref{lem:splits-displayed}.
\end{proof}

\subsection{Putting everything together}
\label{sec:puttingTogetherAncilla}

We now combine all of the results in this section in order to prove \Cref{thm:ancillaVuln}, which is restated below for convenience.

\thmAncillaVuln*
\begin{proof}
    By \Cref{lem:reductionAncilla}, we can compute the $1$-ancilla integrity of an $n$-vertex graph by $n+1$ instances of computing the flip-integrity of an $n$-vertex graph. So in this step we gain an extra factor of $n$. Now we consider computing the flip-integrity of an $n$-vertex graph $G$.
    
    By \Cref{lem:makeGConn}, it suffices to compute the flip-integrity of the largest component $H$ of $G$. We now apply the algorithms in \Cref{lem:constrainedFlipVul} and \Cref{thm:compute-most-balanced-split} to $H$. Thus, in time $\mathcal{O}(n^5)$, we find a set $S \in \mathcal{S}_{\geq 3}(H)$ which minimizes the maximum size of a component of $H \circ S$. We also find a most balanced split $(A,B)$ of $H$ in time $\mathcal{O}(n^2)$. Let $S'$ be the set of all vertices in $A$ which have a neighbor in $B$ and all vertices in $B$ which have a neighbor in $A$. We return the minimum maximum component size of the two graphs $H \circ S$ and $H \circ S'$.
    
    Now, consider an optimal flip set $S^*$ of $H$. Note that $H \circ S^*$ has at least two components by \Cref{lem:lessThanN}. If $H \circ S^*$ has at least three components, then we have found a flip which is at least as good since we considered $H \circ S$. If $H \circ S^*$ has exactly two components, we let $A^*$ and $B^*$ denote their vertex sets. Then $(A^*, B^*)$ is a split of $H$, and so it has width at least the width of $(A,B)$. This means that $\max(|A|, |B|)\leq \max(|A^*|, |B^*|)$. Thus we have found a flip which is at least as good since every component of $H \circ S'$ has size at most $\max(|A|, |B|)$. This completes the proof of \Cref{thm:ancillaVuln}. 
\end{proof}

\section{Acknowledgments}
Nathan Claudet acknowledges funding from the Austrian Science Fund (FWF) [SFB BeyondC F7102, DOI: 10.55776/F71]. For open access purposes, the authors have applied a CC BY public copyright license to any author accepted manuscript version arising from this submission.
Rose McCarty was partially supported by the National Science Foundation under Grant No. DMS-2452111.
Blair D. Sullivan gratefully acknowledges partial financial support for this research by the Fulbright Program, which is sponsored by the U.S. Department of State and the Franco-American Commission -- Fulbright France. The contents of this work are solely the responsibility of the authors and do not necessarily represent the official views of the Fulbright Program, the Government of the United States, or the Franco-American Commission.

\bibliography{integrity}

@article{Hein04,
	author = {Marc Hein and Jens Eisert and Hans J. Briegel},
	doi = {10.1103/physreva.69.062311},
	eprint = {quant-ph/0307130},
	journal = {Physical Review A},
	month = {Jun},
	number = {6},
	publisher = {American Physical Society ({APS})},
	title = {Multiparty entanglement in graph states},
	volume = {69},
	year = 2004
}

@article{Hein06,
	author = {Hein, Marc and D{\"u}r, Wolfgang and Eisert, Jens and Raussendorf, Robert and Maarten Van den Nest and Briegel, Hans J.},
	doi = {10.3254/978-1-61499-018-5-115},
	eprint = {quant-ph/0602096},
	journal = {Quantum computers, algorithms and chaos},
	month = {Mar},
	title = {Entanglement in Graph States and its Applications},
	volume = {162},
	year = {2006}
}

@article{foliagePartition,
    title = {The Foliage Partition: An Easy-to-Compute {LC}-Invariant for Graph States},
    journal = {Quantum},
    volume = {9},
    pages = {1720},
    year = {2025},
    issn = {2521-327X},
    doi = {	https://doi.org/10.22331/q-2025-04-24-1720},
    url = {https://quantum-journal.org/papers/q-2025-04-24-1720/},
    author = {Adam Burchardt and Frederik Hahn}
}

@article{VandenNest04,
	author = {Maarten Van den Nest and Jeroen Dehaene and Bart De Moor},
	doi = {10.1103/physreva.69.022316},
	eprint = {quant-ph/0308151},
	journal = {Physical Review A},
	month = {Feb},
	number = {2},
	publisher = {American Physical Society ({APS})},
	title = {Graphical description of the action of local {C}lifford transformations on graph states},
	volume = {69},
	year = {2004}
}

@book{MiningDatasets,
    author = {Leskovec, Jure and Rajaraman, Anand and Ullman, Jeffrey David},
    title = {Mining of Massive Datasets},
    year = {2014},
    isbn = {1107077230},
    publisher = {Cambridge University Press},
    address = {USA},
    edition = {2nd}
}

@article{davies2025preparing,
  title={Preparing graph states forbidding a vertex-minor},
  author={Davies, James and Jena, Andrew},
  journal={arXiv preprint arXiv:2504.00291},
  year={2025},
  url={https://arxiv.org/abs/2504.00291}
}

@article{transductionSurvey26,
  title={On first-order definable operations on relational structures},
  author={Bruno Courcelle},
  journal={arXiv preprint arXiv:2605.31260},
  year={2026}
}

@article{bravyi2024generating,
  title={Generating $ k $ EPR-pairs from an $ n $-party resource state},
  author={Bravyi, Sergey and Sharma, Yash and Szegedy, Mario and De Wolf, Ronald},
  journal={Quantum},
  volume={8},
  pages={1348},
  year={2024},
  publisher={Verein zur F{\"o}rderung des Open Access Publizierens in den Quantenwissenschaften}
}

@Article{DDV16,
    author={Drange, P{\aa}l Gr{\o}n{\aa}s
    and Dregi, Markus
    and van 't Hof, Pim},
    title={On the Computational Complexity of Vertex Integrity and Component Order Connectivity},
    journal={Algorithmica},
    year={2016},
    month={Dec},
    day={01},
    volume={76},
    number={4},
    pages={1181-1202},
    issn={1432-0541},
    doi={10.1007/s00453-016-0127-x},
    url={https://doi.org/10.1007/s00453-016-0127-x}
}

@misc{BPPSS25,
    title={Low rank {MSO}}, 
    author={Mikołaj Bojańczyk and Michał Pilipczuk and Wojciech Przybyszewski and Marek Sokołowski and Giannos Stamoulis},
    year={2025},
    eprint={2502.08476},
    archivePrefix={arXiv},
    primaryClass={cs.LO},
    url={https://arxiv.org/abs/2502.08476}, 
}

@misc{LS,
  author = {Lenstra, Jan Karel and Shmoys, David B.},
  title  = {The {E}lements of {S}cheduling, {C}hapter 8: {M}inmax {C}riteria without {P}reemption}
}

@article{P96,
    author={Peeters, Ren{\'e}},
    title={Orthogonal representations over finite fields and the chromatic number of graphs},
    journal={Combinatorica},
    year={1996},
    month={Sep},
    day={01},
    volume={16},
    number={3},
    pages={417-431},
    issn={1439-6912},
    doi={10.1007/BF01261326},
    url={https://doi.org/10.1007/BF01261326}
}

@inproceedings{campbell2026erdHos,
  title={The {E}rd{\H{o}}s-{P}{\'o}sa property for circle graphs as vertex-minors},
  author={Campbell, Rutger and Gollin, J Pascal and Hatzel, Meike and Kwon, O-joung and McCarty, Rose and Oum, Sang-il and Wiederrecht, Sebastian},
  booktitle={Proceedings of the 2026 Annual ACM-SIAM Symposium on Discrete Algorithms (SODA)},
  pages={4930--4952},
  year={2026},
  organization={SIAM}
}

@phdthesis{mccarty2021local,
  title={Local structure for vertex-minors},
  author={McCarty, Rose},
  school={University of Waterloo},
  type={PhD thesis},
  address={Waterloo, Ontario, Canada},
  year={2021},
  url={https://uwspace.uwaterloo.ca/items/1cfbfc52-2e30-44a4-b3fb-28493c3d94f0}  
}

@article{2010editDistSurvey,
  title={A survey of graph edit distance},
  author={Gao, Xinbo and Xiao, Bing and Tao, Dacheng and Li, Xuelong},
  journal={Pattern Analysis and Applications},
  volume={13},
  number={1},
  pages={113--129},
  year={2010},
  doi={10.1007/s10044-008-0141-y}
}

@article{2023editDistSurvey,
    title = {A survey of parameterized algorithms and the complexity of edge modification},
    journal = {Computer Science Review},
    volume = {48},
    pages = {100556},
    year = {2023},
    issn = {1574-0137},
    doi = {https://doi.org/10.1016/j.cosrev.2023.100556},
    url = {https://www.sciencedirect.com/science/article/pii/S1574013723000230},
    author = {Christophe Crespelle and Pål Grønås Drange and Fedor V. Fomin and Petr Golovach}
}

@article{transformingStates,
	author = {Dahlberg, Axel  and Wehner, Stephanie},
	title = {Transforming graph states using single-qubit operations},
	journal = {Philosophical Transactions of the Royal Society A: Mathematical, Physical and Engineering Sciences},
	volume = {376},
	number = {2123},
	pages = {20170325},
	year = {2018},
	doi = {10.1098/rsta.2017.0325},
	URL = {https://royalsocietypublishing.org/doi/abs/10.1098/rsta.2017.0325},
	eprint = {https://royalsocietypublishing.org/doi/pdf/10.1098/rsta.2017.0325}
}

@article {BouchetCircleChar,
    AUTHOR = {Bouchet, André},
     TITLE = {Circle graph obstructions},
   JOURNAL = {J. Combin. Theory Ser. B},
  FJOURNAL = {Journal of Combinatorial Theory. Series B},
    VOLUME = {60},
      YEAR = {1994},
    NUMBER = {1},
     PAGES = {107--144},
      ISSN = {0095-8956},
       DOI = {10.1006/jctb.1994.1008},
       URL = {https://doi-org.proxy.lib.uwaterloo.ca/10.1006/jctb.1994.1008},
}

@article{RMgrid,
    title = {The Grid Theorem for vertex-minors},
    journal = {J. Combin. Theory Ser. B},
    volume = {158},
    pages = {93-116},
    year = {2023},
    issn = {0095-8956},
    doi = {https://doi.org/10.1016/j.jctb.2020.08.004},
    url = {https://www.sciencedirect.com/science/article/pii/S0095895620300794},
    author = {Jim Geelen and O-joung Kwon and Rose McCarty and Paul Wollan}
}

@article {graphicIsoSystems,
    AUTHOR = {Bouchet, André},
     TITLE = {Graphic presentations of isotropic systems},
   JOURNAL = {J. Combin. Theory Ser. B},
  FJOURNAL = {Journal of Combinatorial Theory. Series B},
    VOLUME = {45},
      YEAR = {1988},
    NUMBER = {1},
     PAGES = {58--76},
      ISSN = {0095-8956},
       DOI = {10.1016/0095-8956(88)90055-X},
       URL = {https://doi.org/10.1016/0095-8956(88)90055-X},
}

@article {graphMinors20WQO,
    AUTHOR = {Robertson, Neil and Seymour, P. D.},
     TITLE = {Graph minors. {XX}. {W}agner's conjecture},
   JOURNAL = {J. Combin. Theory Ser. B},
  FJOURNAL = {Journal of Combinatorial Theory. Series B},
    VOLUME = {92},
      YEAR = {2004},
    NUMBER = {2},
     PAGES = {325--357},
      ISSN = {0095-8956},
       DOI = {10.1016/j.jctb.2004.08.001},
       URL = {https://doi-org.proxy.lib.uwaterloo.ca/10.1016/j.jctb.2004.08.001},
}

@incollection {deltaMatroidsSurvey,
    AUTHOR = {Moffatt, Iain},
     TITLE = {Delta-matroids for graph theorists},
 BOOKTITLE = {Surveys in combinatorics 2019},
    SERIES = {London Math. Soc. Lecture Note Ser.},
    VOLUME = {456},
     PAGES = {167--220},
 PUBLISHER = {Cambridge Univ. Press, Cambridge},
      YEAR = {2019},
      ISBN = {978-1-108-74072-2}
}

@book{BM08,
  author    = {Bondy, J. A. and Murty, U. S. R.},
  title     = {Graph Theory},
  series    = {Graduate Texts in Mathematics},
  volume    = {244},
  publisher = {Springer},
  year      = {2008},
  doi       = {10.1007/978-1-84628-970-5}
}

@article {RWAndVM,
    AUTHOR = {Sang-il Oum},
     TITLE = {Rank-width and vertex-minors},
   JOURNAL = {J. Combin. Theory Ser. B},
  FJOURNAL = {Journal of Combinatorial Theory. Series B},
    VOLUME = {95},
      YEAR = {2005},
    NUMBER = {1},
     PAGES = {79--100},
      ISSN = {0095-8956},
       DOI = {10.1016/j.jctb.2005.03.003},
       URL = {https://doi.org/10.1016/j.jctb.2005.03.003},
}

@incollection {connectivityIsotropic,
    AUTHOR = {Bouchet, André},
     TITLE = {Connectivity of isotropic systems},
 BOOKTITLE = {Combinatorial {M}athematics: {P}roceedings of the {T}hird
              {I}nternational {C}onference ({N}ew {Y}ork, 1985)},
    SERIES = {Ann. New York Acad. Sci.},
    VOLUME = {555},
     PAGES = {81--93},
 PUBLISHER = {New York Acad. Sci., New York},
      YEAR = {1989},
       DOI = {10.1111/j.1749-6632.1989.tb22439.x},
       URL = {https://doi.org/10.1111/j.1749-6632.1989.tb22439.x},
}

@INPROCEEDINGS{SGT23,
  author={Sen, Aniruddha and Goodenough, Kenneth and Towsley, Don},
  booktitle={2023 IEEE International Conference on Quantum Computing and Engineering (QCE)}, 
  title={Multipartite Entanglement in Quantum Networks Using Subgraph Complementations}, 
  year={2023},
  volume={02},
  number={},
  pages={252-253},
  doi={10.1109/QCE57702.2023.10229}
}

@INPROCEEDINGS {flipWidth23,
    author = {Szymon Toru{\'{n}}czyk},
    booktitle = {2023 IEEE 64th Annual Symposium on Foundations of Computer Science (FOCS)},
    title = {Flip-width: Cops and Robber on dense graphs},
    year = {2023},
    volume = {},
    issn = {},
    pages = {663-700},
    doi = {10.1109/FOCS57990.2023.00045},
    url = {https://doi.ieeecomputersociety.org/10.1109/FOCS57990.2023.00045},
    publisher = {IEEE Computer Society},
    address = {Los Alamitos, CA, USA},
    month = {nov}
}

@INPROCEEDINGS {flipperGamesMonStable23,
    author = {Gajarsk{\'{y}}, Jakub and M{\"{a}}hlmann, Nikolas and McCarty, Rose and Ohlmann, Pierre and Pilipczuk, Micha{\l} and Przybyszewski, Wojciech and Siebertz, Sebastian and Soko{\l}owski, Marek and Toru{\'{n}}czyk, Szymon},
    title =	{{Flipper Games for Monadically Stable Graph Classes}},
    booktitle =	{50th International Colloquium on Automata, Languages, and Programming (ICALP)},
    pages =	{128:1--128:16},
    series =	{Leibniz International Proceedings in Informatics (LIPIcs)},
    ISBN =	{978-3-95977-278-5},
    ISSN =	{1868-8969},
    year =	{2023},
    volume =	{261},
    editor =	{Etessami, Kousha and Feige, Uriel and Puppis, Gabriele},
    publisher =	{Schloss Dagstuhl -- Leibniz-Zentrum f{\"u}r Informatik},
    address =	{Dagstuhl, Germany},
    URL =		{https://drops.dagstuhl.de/entities/document/10.4230/LIPIcs.ICALP.2023.128},
    URN =		{urn:nbn:de:0030-drops-181804},
    doi =		{10.4230/LIPIcs.ICALP.2023.128}
}

@inproceedings{mergeWidth25,
    author = {Dreier, Jan and Toru{\'{n}}czyk, Szymon},
    title = {Merge-Width and First-Order Model Checking},
    year = {2025},
    isbn = {9798400715105},
    publisher = {Association for Computing Machinery},
    address = {New York, NY, USA},
    url = {https://doi.org/10.1145/3717823.3718259},
    doi = {10.1145/3717823.3718259},
    booktitle = {Proceedings of the 57th Annual ACM Symposium on Theory of Computing},
    pages = {1944–1955},
    numpages = {12},
    location = {Prague, Czechia},
    series = {STOC '25}
}

@article {partialComplementFGST20,
    AUTHOR = {Fomin, Fedor V. and Golovach, Petr A. and Str\o mme, Torstein
              J. F. and Thilikos, Dimitrios M.},
     TITLE = {Subgraph complementation},
   JOURNAL = {Algorithmica},
  FJOURNAL = {Algorithmica. An International Journal in Computer Science},
    VOLUME = {82},
      YEAR = {2020},
    NUMBER = {7},
     PAGES = {1859--1880},
      ISSN = {0178-4617,1432-0541},
   MRCLASS = {05C85 (05C75 68Q17 68Q25 68R10)},
  MRNUMBER = {4099974},
       DOI = {10.1007/s00453-020-00677-8},
       URL = {https://doi.org/10.1007/s00453-020-00677-8},
}

@article {partialComplHFree22,
    AUTHOR = {Antony, Dhanyamol and Garchar, Jay and Pal, Sagartanu and
              Sandeep, R. B. and Sen, Sagnik and Subashini, R.},
     TITLE = {On subgraph complementation to {$H$}-free graphs},
   JOURNAL = {Algorithmica},
  FJOURNAL = {Algorithmica. An International Journal in Computer Science},
    VOLUME = {84},
      YEAR = {2022},
    NUMBER = {10},
     PAGES = {2842--2870},
      ISSN = {0178-4617,1432-0541},
   MRCLASS = {68R10},
  MRNUMBER = {4491029},
       DOI = {10.1007/s00453-022-00991-3},
       URL = {https://doi.org/10.1007/s00453-022-00991-3},
}

@article {partialComplAPS25,
    AUTHOR = {Antony, Dhanyamol and Pal, Sagartanu and Sandeep, R. B.},
     TITLE = {Algorithms for subgraph complementation to some classes of
              graphs},
   JOURNAL = {Inform. Process. Lett.},
  FJOURNAL = {Information Processing Letters},
    VOLUME = {188},
      YEAR = {2025},
     PAGES = {Paper No. 106530, 5},
      ISSN = {0020-0190,1872-6119},
   MRCLASS = {68R10 (68Q27)},
  MRNUMBER = {4785291},
       DOI = {10.1016/j.ipl.2024.106530},
       URL = {https://doi.org/10.1016/j.ipl.2024.106530},
}

@article {partialComplHardAPS24,
    AUTHOR = {Antony, Dhanyamol and Pal, Sagartanu and Sandeep, R. B. and
              Subashini, R.},
     TITLE = {Cutting a tree with subgraph complementation is hard, except
              for some small trees},
   JOURNAL = {J. Graph Theory},
  FJOURNAL = {Journal of Graph Theory},
    VOLUME = {107},
      YEAR = {2024},
    NUMBER = {1},
     PAGES = {126--168},
      ISSN = {0364-9024,1097-0118},
   MRCLASS = {05C85 (05C05 05C75)},
  MRNUMBER = {4788444},
MRREVIEWER = {Jesper\ Jansson},
       DOI = {10.1002/jgt.23112},
       URL = {https://doi.org/10.1002/jgt.23112},
}

@incollection {GroheSimilarity25,
    AUTHOR = {Grohe, Martin},
     TITLE = {Some thoughts on graph similarity},
 BOOKTITLE = {Principles of verification: cycling the probabilistic
              landscape---essays dedicated to {J}oost-{P}ieter {K}atoen on
              the occasion of his 60th birthday. {P}art {I}},
    SERIES = {Lecture Notes in Comput. Sci.},
    VOLUME = {15260},
     PAGES = {369--392},
 PUBLISHER = {Springer, Cham},
      YEAR = {2025},
      ISBN = {978-3-031-75782-2; 978-3-031-75783-9},
   MRCLASS = {68R10},
  MRNUMBER = {4864281},
       DOI = {10.1007/978-3-031-75783-9\_15},
       URL = {https://doi.org/10.1007/978-3-031-75783-9_15},
}

@article{BRpersistency2001,
  title = {Persistent Entanglement in Arrays of Interacting Particles},
  author = {Briegel, Hans J. and Raussendorf, Robert},
  journal = {Phys. Rev. Lett.},
  volume = {86},
  issue = {5},
  pages = {910--913},
  numpages = {0},
  year = {2001},
  month = {Jan},
  publisher = {American Physical Society},
  doi = {10.1103/PhysRevLett.86.910},
  url = {https://link.aps.org/doi/10.1103/PhysRevLett.86.910}
}

@article {BPRminrank22,
    AUTHOR = {Buchanan, Calum and Purcell, Christopher and Rombach, Puck},
     TITLE = {Subgraph complementation and minimum rank},
   JOURNAL = {Electron. J. Combin.},
  FJOURNAL = {Electronic Journal of Combinatorics},
    VOLUME = {29},
      YEAR = {2022},
    NUMBER = {1},
     PAGES = {Paper No. 1.38, 20},
      ISSN = {1077-8926},
   MRCLASS = {05C62 (05C50 05C75)},
  MRNUMBER = {4395941},
       DOI = {10.37236/10383},
       URL = {https://doi.org/10.37236/10383},
}

@article {Aminrank75,
    AUTHOR = {Lempel, Abraham},
     TITLE = {Matrix factorization over {${\rm GF}(2)$} and trace-orthogonal
              bases of {${\rm GF}(2\sp{n})$}},
   JOURNAL = {SIAM J. Comput.},
  FJOURNAL = {SIAM Journal on Computing},
    VOLUME = {4},
      YEAR = {1975},
     PAGES = {175--186},
      ISSN = {0097-5397},
   MRCLASS = {15A33},
  MRNUMBER = {376715},
MRREVIEWER = {A.\ D.\ Porter},
       DOI = {10.1137/0204014},
       URL = {https://doi.org/10.1137/0204014},
}

@article{vdNDVB07,
  title = {Classical simulation versus universality in measurement-based quantum computation},
  author = {Van den Nest, M. and D{\"u}r, W. and Vidal, G. and Briegel, H. J.},
  journal = {Phys. Rev. A},
  volume = {75},
  issue = {1},
  pages = {012337},
  numpages = {15},
  year = {2007},
  month = {Jan},
  publisher = {American Physical Society},
  doi = {10.1103/PhysRevA.75.012337},
  url = {https://link.aps.org/doi/10.1103/PhysRevA.75.012337}
}

@article {FGP2020,
    AUTHOR = {Fomin, Fedor V. and Golovach, Petr A. and Panolan, Fahad},
     TITLE = {Parameterized low-rank binary matrix approximation},
   JOURNAL = {Data Min. Knowl. Discov.},
  FJOURNAL = {Data Mining and Knowledge Discovery},
    VOLUME = {34},
      YEAR = {2020},
    NUMBER = {2},
     PAGES = {478--532},
      ISSN = {1384-5810,1573-756X},
   MRCLASS = {62H30},
  MRNUMBER = {4064386},
       DOI = {10.1007/s10618-019-00669-5},
       URL = {https://doi.org/10.1007/s10618-019-00669-5},
}

@article {MMS16,
    AUTHOR = {Meesum, S. M. and Misra, Pranabendu and Saurabh, Saket},
     TITLE = {Reducing rank of the adjacency matrix by graph modification},
   JOURNAL = {Theoret. Comput. Sci.},
  FJOURNAL = {Theoretical Computer Science},
    VOLUME = {654},
      YEAR = {2016},
     PAGES = {70--79},
      ISSN = {0304-3975,1879-2294},
   MRCLASS = {05C50 (05C75 05C85 68Q25)},
  MRNUMBER = {3574485},
       DOI = {10.1016/j.tcs.2016.02.020},
       URL = {https://doi.org/10.1016/j.tcs.2016.02.020},
}

@article{PHSH26,
  title={Bipartitioning of Graph States for Distributed Measurement-Based Quantum Computing},
  author={Kjell Fredrik Pettersen and Matthias Heller and Giorgio Sartor and Raoul Heese},
  journal={arXiv preprint arXiv2601.06332},
  year={2026},
  url={https://arxiv.org/abs/2601.06332}
}

@article {photonicLinearRW,
    AUTHOR = {Li, Bikun and Economou, Sophia E. and Barnes, Edwin},
     TITLE = {Photonic resource state generation from a minimal number of quantum emitters},
   JOURNAL = {npj Quantum Information},
    VOLUME = {8},
    issue = {1},
      YEAR = {2022},
       DOI = {10.1038/s41534-022-00522-6},
       URL = {https://doi.org/10.1038/s41534-022-00522-6},
}

@article {cunningham82,
    AUTHOR = {Cunningham, William H.},
     TITLE = {Decomposition of directed graphs},
   JOURNAL = {SIAM J. Algebraic Discrete Methods},
  FJOURNAL = {Society for Industrial and Applied Mathematics. Journal on
              Algebraic and Discrete Methods},
    VOLUME = {3},
      YEAR = {1982},
    NUMBER = {2},
     PAGES = {214--228},
      ISSN = {0196-5212},
   MRCLASS = {05C20 (05C70 68E10)},
  MRNUMBER = {655562},
       DOI = {10.1137/0603021},
       URL = {https://doi.org/10.1137/0603021},
}

@article {spinrad89,
    AUTHOR = {Spinrad, Jeremy},
     TITLE = {Prime testing for the split decomposition of a graph},
   JOURNAL = {SIAM J. Discrete Math.},
  FJOURNAL = {SIAM Journal on Discrete Mathematics},
    VOLUME = {2},
      YEAR = {1989},
    NUMBER = {4},
     PAGES = {590--599},
      ISSN = {0895-4801},
   MRCLASS = {68Q25 (05C75 68R10)},
  MRNUMBER = {1018541},
MRREVIEWER = {Lawrence\ V.\ Saxton},
       DOI = {10.1137/0402051},
       URL = {https://doi.org/10.1137/0402051},
}

@article {linearSplitRevisited12,
    AUTHOR = {Charbit, Pierre and de Montgolfier, Fabien and Raffinot,
              Mathieu},
     TITLE = {Linear time split decomposition revisited},
   JOURNAL = {SIAM J. Discrete Math.},
  FJOURNAL = {SIAM Journal on Discrete Mathematics},
    VOLUME = {26},
      YEAR = {2012},
    NUMBER = {2},
     PAGES = {499--514},
      ISSN = {0895-4801,1095-7146},
   MRCLASS = {68R10 (05C51 05C85)},
  MRNUMBER = {2967479},
MRREVIEWER = {Matthew\ Yancey},
       DOI = {10.1137/10080052X},
       URL = {https://doi.org/10.1137/10080052X},
}

@misc{KT26,
      title={Connectivity augmentation is fixed-parameter tractable}, 
      author={Tuukka Korhonen and Mikkel Thorup},
      year={2026},
      eprint={2605.11757},
      archivePrefix={arXiv},
      primaryClass={cs.DS},
      url={https://arxiv.org/abs/2605.11757}, 
}

@inbook{CR26,
author = {Johannes Carmesin and M. S. Ramanujan},
title = {Augmenting to 4-vertex connectivity is fixed-parameter tractable},
booktitle = {Proceedings of the 2026 Annual ACM-SIAM Symposium on Discrete Algorithms (SODA)},
chapter = {},
pages = {2015-2042},
doi = {10.1137/1.9781611978971.73},
URL = {https://epubs.siam.org/doi/abs/10.1137/1.9781611978971.73},
eprint = {https://epubs.siam.org/doi/pdf/10.1137/1.9781611978971.73}, 
year = {2026}
}

@misc{PilipczukSurvey26,
      title={Graph classes through the lens of logic}, 
      author={Micha{\l} Pilipczuk},
      year={2026},
      eprint={2501.04166},
      archivePrefix={arXiv},
      primaryClass={math.CO},
      url={https://arxiv.org/abs/2501.04166}, 
}

@article{linearSplitDahlhaus,
    title = {Parallel Algorithms for Hierarchical Clustering and Applications to Split Decomposition and Parity Graph Recognition},
    journal = {Journal of Algorithms},
    volume = {36},
    number = {2},
    pages = {205-240},
    year = {2000},
    issn = {0196-6774},
    doi = {https://doi.org/10.1006/jagm.2000.1090},
    url = {https://www.sciencedirect.com/science/article/pii/S0196677400910903},
    author = {Elias Dahlhaus}
}

@misc{OS26,
      title={Polynomial-size encoding of all cuts of small value in integer-valued symmetric submodular functions}, 
      author={Sang-il Oum and Marek Soko{\l}owski},
      year={2026},
      eprint={2603.10710},
      archivePrefix={arXiv},
      primaryClass={math.CO},
      url={https://arxiv.org/abs/2603.10710}, 
}

@inproceedings{BGW25,
    author = {Bouland, Adam and Giurgic{\u{a}}-Tiron, Tudor and Wright, John},
    title = {The State Hidden Subgroup Problem and an Efficient Algorithm for Locating Unentanglement},
    year = {2025},
    isbn = {9798400715105},
    publisher = {Association for Computing Machinery},
    address = {New York, NY, USA},
    url = {https://doi.org/10.1145/3717823.3718118},
    doi = {10.1145/3717823.3718118},
    booktitle = {Proceedings of the 57th Annual ACM Symposium on Theory of Computing},
    pages = {463–470},
    numpages = {8},
    location = {Prague, Czechia},
    series = {STOC '25}
}

@incollection{FonDerFlaass1988,
	address = {Amsterdam},
	author = {Fon-Der-Flaass, Dmitri G.},
	booktitle = {Combinatorics (Eger, 1987)},
	mrclass = {05C45},
	mrnumber = {94g:05052},
	mrreviewer = {Andr{\'e} Bouchet},
	pages = {257--266},
	publisher = {North-Holland},
	series = {Colloq. Math. Soc. J\'anos Bolyai},
	title = {On local complementations of graphs},
	volume = {52},
	year = {1988}
}

@article{briegel2009measurement,
	author = {Briegel, Hans J. and Browne, David E. and D{\"u}r, Wolfgang and Raussendorf, Robert and Maarten Van den Nest},
	journal = {Nature Physics},
	number = {1},
	pages = {19--26},
	publisher = {Nature Publishing Group UK London},
	title = {Measurement-based quantum computation},
	doi={10.1038/nphys1157},
	volume = {5},
	year = {2009},
	eprint={0910.1116}
}

@article{raussendorf2001one,
	author = {Raussendorf, Robert and Briegel, Hans J.},
	journal = {Physical Review Letters},
	number = {22},
	pages = {5188},
	publisher = {APS},
	title = {A one-way quantum computer},
	doi={10.1103/PhysRevLett.86.5188},
	volume = {86},
	year = {2001}
}

@article{raussendorf2003measurement,
	author = {Raussendorf, Robert and Browne, Daniel E. and Briegel, Hans J.},
	journal = {Physical review A},
	number = {2},
	pages = {022312},
	publisher = {APS},
	title = {Measurement-based quantum computation on cluster states},
	eprint={quant-ph/0301052},
	doi={10.1103/PhysRevA.68.022312},
	volume = {68},
	year = {2003}
}

@article{Hahn2022limitations,
  title = {Limitations of nearest-neighbor quantum networks},
  author = {Hahn, F. and Dahlberg, A. and Eisert, J. and Pappa, A.},
  journal = {Phys. Rev. A},
  volume = {106},
  issue = {1},
  pages = {L010401},
  numpages = {5},
  year = {2022},
  month = {Jul},
  publisher = {American Physical Society},
  doi = {10.1103/PhysRevA.106.L010401},
  url = {https://link.aps.org/doi/10.1103/PhysRevA.106.L010401}
}

@phdthesis{hahn2022phdthesis,
  title={Quantum Networks},
  author={Hahn, Frederik},
  year={2022}, 
  url = "http://dx.doi.org/10.17169/refubium-41101",
  school = {Freie Universität Berlin}
}

@inproceedings{Cautres2024,
	author = {Cautr\`{e}s, Maxime and Claudet, Nathan and Mhalla, Mehdi and Perdrix, Simon and Savin, Valentin and Thomass\'{e}, St\'{e}phan},
	title =	{Vertex-Minor Universal Graphs for Generating Entangled Quantum Subsystems},
	booktitle = {Proceedings of the 51st International Colloquium on Automata, Languages, and Programming (ICALP 2024)},
	doi = {10.4230/LIPIcs.ICALP.2024.36},
	year = {2024},
	eprint={2402.06260}
}

@article{markham2008graph,
	author = {Markham, Damian and Sanders, Barry C.},
	journal = {Physical Review A},
	number = {4},
	pages = {042309},
	publisher = {APS},
	title = {Graph states for quantum secret sharing},
	doi={10.1103/PhysRevA.78.042309},
	volume = {78},
	year = {2008},
	eprint={0808.1532}
}

@article{Keet2010,
  title = {Quantum secret sharing with qudit graph states},
  author = {Keet, Adrian and Fortescue, Ben and Markham, Damian and Sanders, Barry C.},
  journal = {Physical Review A},
  volume = {82},
  issue = {6},
  pages = {062315},
  numpages = {11},
  year = {2010},
  month = {Dec},
  publisher = {American Physical Society},
  doi = {10.1103/PhysRevA.82.062315},
  eprint={1004.4619}
}

@InProceedings{Javelle2013,
author="Javelle, J{\'e}r{\^o}me
and Mhalla, Mehdi
and Perdrix, Simon",
title="New Protocols and Lower Bounds for Quantum Secret Sharing with Graph States",
booktitle="Proceedings of the 7th Conference on the Theory of Quantum Computation, Communication, and Cryptography (TQC 2012)",
year="2013",
pages="1--12", 
eprint={1109.1487},
doi={10.1007/978-3-642-35656-8_1}
}

@inproceedings{gravier2013quantum,
	author = {Gravier, Sylvain and Javelle, J{\'e}r{\^o}me and Mhalla, Mehdi and Perdrix, Simon},
	booktitle = {Proceedings of the 8th {M}athematical and {E}ngineering {M}ethods in {C}omputer {S}cience International Doctoral Workshop ({MEMICS} 2012)},
	title = {Quantum secret sharing with graph states},
	doi={10.1007/978-3-642-36046-6_3},
	year = {2013},
	url={https://hal.science/hal-00933722/document}
}

@Article{Bell2014secret,
author={Bell, B. A.
and Markham, Damian
and Herrera-Mart{\'i}, D. A.
and Marin, Anne
and Wadsworth, W. J.
and Rarity, J. G.
and Tame, M. S.},
title={Experimental demonstration of graph-state quantum secret sharing},
journal={Nature Communications},
year={2014},
month={Nov},
day={21},
volume={5},
number={1},
pages={5480},
issn={2041-1723},
doi={10.1038/ncomms6480},
url={https://doi.org/10.1038/ncomms6480},
eprint={1411.5827}
}

@article{hahn2019quantum,
	author = {Hahn, Frederik and Pappa, Anna and Eisert, Jens},
	doi = {10.1038/s41534-019-0191-6},
	eprint = {1805.04559},
	journal = {npj Quantum Information},
	number = {1},
	pages = {1--7},
	publisher = {Nature Publishing Group},
	title = {Quantum network routing and local complementation},
	volume = {5},
	year = {2019}
}

@article{meignant2019distributing,
	author = {Meignant, Cl\'ement and Markham, Damian and Grosshans, Fr\'ed\'eric},
	doi = {10.1103/PhysRevA.100.052333},
	eprint = {1811.05445},
	issue = {5},
	journal = {Physical Review A},
	month = {Nov},
	numpages = {6},
	pages = {052333},
	publisher = {American Physical Society},
	title = {Distributing graph states over arbitrary quantum networks},
	volume = {100},
	year = {2019}
}

@inproceedings{fischer2021distributing,
	author = {Fischer, Alex and Towsley, Don},
	booktitle = {Proceedings of the 2021 IEEE International Conference on Quantum Computing and Engineering (QCE)},
	doi = {10.1109/QCE52317.2021.00049},
	eprint = {2009.10888},
	pages = {324--333},
	title = {Distributing graph states across quantum networks},
	year = {2021}
}

@article{Mannalath2023,
  title = {Multiparty entanglement routing in quantum networks},
  author = {Mannalath, Vaisakh and Pathak, Anirban},
  journal = {Physical Review A},
  volume = {108},
  issue = {6},
  pages = {062614},
  numpages = {15},
  year = {2023},
  month = {Dec},
  publisher = {American Physical Society},
  doi = {10.1103/PhysRevA.108.062614},
  eprint={2211.06690}
}

@article{freund2025graph,
  title={Graph state extraction from two-dimensional cluster states},
  author={Freund, Julia and Pirker, Alexander and Vandr{\'e}, Lina and D{\"u}r, Wolfgang},
  journal={New Journal of Physics},
  volume={27},
  number={9},
  pages={094505},
  year={2025},
  publisher={IOP Publishing}
}

@article{vdn2006universal,
  title = {Universal Resources for Measurement-Based Quantum Computation},
  author = {Van den Nest, Maarten and Miyake, Akimasa and D\"ur, Wolfgang and Briegel, Hans J.},
  journal = {Phys. Rev. Lett.},
  volume = {97},
  issue = {15},
  pages = {150504},
  numpages = {4},
  year = {2006},
  month = {Oct},
  publisher = {American Physical Society},
  doi = {10.1103/PhysRevLett.97.150504},
  url = {https://link.aps.org/doi/10.1103/PhysRevLett.97.150504}
}

@inproceedings{Javelle12,
	author = {J{\'e}r{\^o}me Javelle and Mehdi Mhalla and Simon Perdrix},
	booktitle = {Proceedings of the 38th workshop on Graph Theory (WG 2012)},
	doi = {10.1007/978-3-642-34611-8_16},
	title = {On the Minimum Degree Up to Local Complementation: Bounds and Complexity},
	year = {2012},
	eprint={1204.4564}
}

@inproceedings{CattaneoP15,
	author = {David Cattan{\'{e}}o and Simon Perdrix},
	booktitle = {Proceedings of the 26th International Symposium on Algorithms and Computation, ({ISAAC} 2015)},
	doi = {10.1007/978-3-662-48971-0\_23},
	title = {Minimum Degree Up to Local Complementation: Bounds, Parameterized Complexity, and Exact Algorithms},
	year = {2015},
	eprint={1503.04702}}

\end{document}